\documentclass[10pt,a4paper,openright,tablecaptionabove]{scrartcl}
\usepackage{subfig}
\usepackage[english]{babel}
\usepackage{amsmath, amsthm, amssymb}
\usepackage{bbm}
\usepackage{url}
\usepackage{tikz-cd}
\usepackage{braket}
\usepackage{path}
\usepackage{nicefrac}
\usepackage{amssymb}
\usepackage{amsmath}
\usepackage{amsthm}
\usepackage{cancel}
\usepackage{indentfirst}
\usepackage{eucal}
\usepackage{appendix}
\usepackage[utf8]{inputenc}
\usepackage{geometry}
\usepackage{verbatim}
\usepackage{nameref}
\usepackage{alltt}
\usepackage{hyperref}
\usepackage{pgfplots}
\usepackage{dsfont}
\usepackage{appendix}
\usepackage[font=small,labelfont=bf]{caption} 

\newtheorem{prop}{Proposition}
\newtheorem{rem}{Remark}

\newcommand{\figref}[1]{\figurename~\ref{#1}}

\newcommand{\bx}{\boldsymbol{{x}}}

\newcommand{\bp}{\boldsymbol{{p}}}

\usetikzlibrary{arrows.meta,automata,positioning,shadows}

\areaset[current]{500pt}{680pt}
\newtheorem{theorem}{Theorem}[section]

\newtheorem{lemma}[theorem]{Lemma}
\newtheorem{definition}{Definition}[section]

\begin{document} 

\numberwithin{equation}{section}

\thispagestyle{empty}
\begin{center}
\Large{{\bf Racah Algebra $R(n)$ from Coalgebraic Structures and Chains of $R(3)$ Substructures}}
\end{center}
\vskip 0.5cm
\begin{center}
\textsc{Danilo Latini, Ian Marquette and Yao-Zhong Zhang}
\end{center}
\begin{center}
School of Mathematics and Physics, The University of Queensland, Brisbane, QLD 4072, Australia
\end{center}
\vskip 0.5cm

\vskip  1cm
\hrule
\begin{center}
\textsf{{\bf abstract}}
\end{center}
\begin{abstract}
	\noindent   The recent interest in the study of higher-rank polynomial algebras related to $n$-dimensional classical and quantum superintegrable systems with coalgebra symmetry  and their connection with the generalised Racah algebra $R(n)$, a higher-rank generalisation of the rank one Racah algebra $R(3)$, raises the problem of understanding the role played by the $n - 2$ quadratic subalgebras generated by the left and right Casimir invariants (sometimes referred as \emph{universal quadratic substructures}) from this new perspective. Such subalgebra structures play a signficant role in the algebraic derivation of spectrum of quantum superintegrable systems. In this work, we tackle this problem and show that the above quadratic subalgebra structures can be understood, at a fixed $n > 3$, as the images of $n - 2$ injective morphisms of $R(3)$ into $R(n)$. We show that each of the $n-2$  quadratic subalgebras is isomorphic to the rank one Racah algebra $R(3)$. As a byproduct, we also obtain an equivalent presentation for the universal quadratic subtructures generated by the partial Casimir invariants of the coalgebra. The construction, which relies on explicit (symplectic or differential) realisations of the generators, is performed in both the classical and the quantum cases. 
\end{abstract}
\vskip 0.35cm
\hrule

%
%
%
%
%
\section{Introduction}
\label{sec1}

Superintegrable systems represent an exceptional subset of integrable Hamiltonian systems. For a classical Hamiltonian system $H=H(x_1, \dots, x_n, p_1, \dots, p_n)$ with $n$ degrees of freedom, the notion of integrability requires the existence of $n$ integrals of motion, i.e. $n$ well-defined functionally independent functions on the phase-space, say  $\{X_1, \dots, X_n\}$, that Poisson commute with the Hamiltonian $X_1:=H$, i.e. $\{X_1, X_i\}=0$ for $i=1, \dots, n$. Moreover, the above constants have to be in involution, so that $\{X_i, X_j\}=0$ for $1 \leq i,j\leq  n$.
The quantum counterpart of these models also exists. In this case, it is required the existence of $n$ \textquotedblleft quantum integrals of motion\textquotedblright, i.e. $n$ well-defined algebraically independent Hermitian operators $\{\hat{X}_1, \dots, \hat{X}_n\}$ on a given Hilbert space, which commute with the quantum Hamiltonian $\hat{H}=\hat{H}(\hat{x}_1, \dots, \hat{x}_n, \hat{p}_1, \dots, \hat{p}_n)$, i.e. $[\hat{X}_1, \hat{X}_i]=0$, with $\hat{X}_1:= \hat{H}$. Moreover, they have to commute pair-wise, meaning that $[\hat{X}_i, \hat{X}_j]=0$ for $1 \leq i,j \leq n$.

 When $k$ functionally/algebraically independent additional classical/quantum integrals of motion arise, i.e. there exists another set $\{Y_1, \dots, Y_k\}$/$\{\hat{Y}_1, \dots, \hat{Y}_k\}$, with $1 \leq k \leq n-1$, such as $\{Y_i, H\}=0$/$[\hat{Y}_i,\hat{H}]=0$ for $i=1, \dots, k$ we come to the definition of classical/quantum superintegrability. In particular, for $k=1$ we deal with \emph{minimal superintegrability} whereas if $k=n-1$ with \emph{maximal superintegrability}\footnote{Here, and throughout the paper, we shall assume the classical/quantum integrals to be polynomials in the momenta/finite-order differential operators, so that we are implicitly defining, following the terminology used in \cite{MillerPostWinternitz2013R}, polynomial superintegrability/superintegrability of finite-order. In this perspective, the order as a polynomial in the momenta/as a linear differential operator defines the order of the classical/quantum integral of motion.}. It is also common to call \emph{quasi-maximally superintegrable} those systems for which just one classical/quantum integral is missing, so that one has a total number of $2n-2$ functionally/algebraically independent classical/quantum integrals (the Hamiltonian included). In the maximally superintegrable case, it is not difficult to imagine that the existence of such a large number of symmetries for a given Hamiltonian system will make it very special, as many unique properties will arise.  For example, finite closed trajectories and periodic motion in classical mechanics or accidental degeneracy of the energy spectrum in quantum mechanics, just to cite a few of them. For a state-of-the-art perspective on superintegrability we refer the reader to \cite{MillerPostWinternitz2013R} where an exhaustive description (mostly focused on the algebraic side) of  superintegrable models and their properties can be found, also in connection to special functions theory, together with several explicit examples. For a more geometrical point of view we refer the reader to \cite{2014RCD....19..415B}.
 
  From the algebraic perspective, the integrals of motion of a $n$D superintegrable system usually close to higher rank polynomial algebras (see for example \cite{Hoque_2015, 1751-8121-49-12-125201, Hoque_2015_, Hoque2018, Chen_2019, Chen_2019_}). This is contrast to the purely integrable case for which all the integrals of motion just define abelian algebras. The richness of unique properties these classical and quantum models show justified the beginning of their classification \cite{Kalnins05_, Kalnins05_2, Kalnins05_3, Kalnins06_4, Kalnins06_5}. Remarkably, for the two-dimensional case, it has been proved that all second-order superintegrable systems can be obtained from a unique model known as \emph{the generic superintegrable system on the 2-sphere} \cite{Kalnins_2013}. A higher dimensional extension of this model, \emph{the generic superintergrable system on the $(n-1)$-sphere}, has been recently studied in a series of papers (see \cite{De_Bie_2017, Iliev_2017, bie2020racah, Gaboriaud_2019, Crampe2020} and references therein) devoted to the so-called generalised Racah algebra $R(n)$, a higher rank generalisation of the rank one Racah algebra $R(3)$.  Interestingly, the authors proved that the above algebraic structure is naturally realised as the symmetry algebra of the above-mentioned superintegrable model.
 
 In a recent paper \cite{latini2020embedding}, we have shown that the above algebraic structure also appears inside a larger symmetry algebra (with additional generators besides the ones associated to the Racah sector) related to two specific Euclidean models with non-central terms, the Smorodinsky-Winternitz system \cite{PhysRevA.41.5666, FRIS1965354, Makarov1967, EVANS1990483, doi:10.1063/1.529449, Ballesteros_2004} and the generalized Kepler-Coulomb system  \cite{doi:10.1063/1.2840465,1751-8121-42-24-245203, 2011SIGMA...7..054T}. More generally, it is realised in classical/quantum mechanics as the quadratic Poisson/associative algebra of  classical/quantum integrals of motion arising for a quasi-maximally superintegrable family of $n$D Euclidean classical/quantum Hamiltonian systems characterised by a radial potential $V=V(r)$  and $n$ additional non-central terms of the type $a_i/x_i^2$ ($i=1, \dots, n$) which break the radial symmetry. Such a family of Hamiltonians is included in the subset of $n$D superintegrable systems that are endowed with an $\mathfrak{sl}(2, \mathbb{R})$ \emph{coalgebra symmetry} \cite{Ballesteros1996, 0305-4470-31-16-009, 1742-6596-175-1-012004}.  For these models, two sets composed by $n-1$  functionally independent integrals of motion can be extracted from the so-called \textquotedblleft left\textquotedblright and \textquotedblleft right\textquotedblright coproducts of the non-linear Casimir of the $\mathfrak{sl}(2,\mathbb{R})$ coalgebra. One of the integrals is in common, so that one has a total number of $(n-1)+(n-1)-1=2n-3$ classical/quantum integrals of motion besides the Hamiltonian (one constant left for maximal superintegrability). These classical/quantum integrals, which arise under a given symplectic/differential realisation, turn out to be functionally/algebraically independent by construction.
 Such an approach to superintegrability, which holds for both classical and quantum systems because of its algebraic nature, has found application to a variety of problems, and many superintegrable systems have been understood in this framework, such as for example systems defined on non Euclidean spaces \cite{2008PhyD..237..505B,  Ballesteros_2008, BALLESTEROS2011, Ballesteros2009, Post_2015}, models with spin-orbital interactions \cite{Riglioni_2014} and discrete quantum mechanical models \cite{LATINI20163445}. 
 
A more recent application, deeply connected to the aim of the present work,  can be found in \cite{Latini_2019}. In that work, the chain structure of quadratic algebras with three generators discovered in \cite{10.1088/1751-8121/aac111} has been rephrased in the framework of coalgebra symmetry. Such algebraic structure has been used to construct algebraically the energy spectrum of a second-order $n$D Euclidean maximally superintegrable system with $n-1$ non-central terms (a quasi-generalised KC system) through deformed oscillator realisations of the chain structures and set of cubic Casimirs.  Such idea of substructures was originally introduced for a three-dimensional system, the nondegenerate (sometimes called generalised) Kepler-Coulomb model \cite{2011SIGMA...7..054T}.  The link with coalgebra symmetry led to the identification of the generators of part of the above chain of quadratic algebras ($n-2$ of the $n-1$ quadratic substructures forming the entire chain) with the left and right Casimir invariants of the coalgebra. As a consequence, the result has been generalised to all $n$D superintegrable models (both classical and quantum) possessing the same $\mathfrak{sl}(2, \mathbb{R})$ coalgebra symmetry. For this reason, we referred to the above chain of quadratic algebras as being \emph{universal}. The full chain of $n-1$ quadratic algebras used in \cite{10.1088/1751-8121/aac111} to construct the spectrum of the $n$D model algebraically also include an additional constant besides the left and right Casimirs. This is related to the specific model investigated. In fact, it was second-order maximally superintegrable and, because of this, admitted an additional integral of motion (a generalisation of the Laplace-Runge-Lenz (LRL) vector \cite{doi:10.1119/1.9745,doi:10.1119/1.10202, goldstein2002classical}). However, here we are only interested to investigate, at a fixed $n \geq 3$, the $n-2$ quadratic structures generated by the above left and right constants. Additional commutation relations appear when maximally superintegrable subcases are considered. 

The connection observed between the coalgebra symmetry and the substructures needs to be further explained. In this paper we plan to connect explicitly these quadratic chain structures of three generators to the Racah algebra. 

Specifically, the aim of this work is to show that the above quadratic   substructures can be understood as being  the  subalgebras of $R(n)$ arising as the images of $n-2$ injective morphisms of $R(3)$ into $R(n)$ and that, as a direct consequence, they turn out to be isomorphic to (a classical/quantum realisation of) the rank one Racah algebra $R(3)$, thus elucidating the deeper algebraic meaning behind the above substructures.
\vskip 0.1cm 
We organize the paper as follows:

\begin{itemize}
	\item In Section \ref{sec2} we recall briefly the general idea behind the coalgebra symmetry approach to superintegrable systems and, in particular, we focus on the definition of the left and right Casimir invariants. This is in order to introduce some background that will be useful to follow the developments of the work. We present the construction for Poisson coalgebras (the extension to noncommutative coalgebras is direct).
	\item In Section \ref{sec3} we specialise to the Lie-Poisson (co)algebra $\mathfrak{sl}(2,\mathbb{R})$ endowed with the (primitive) coassociative coproduct $\Delta$. After obtaining the corresponding left and right Casimir invariants, we elucidate their connection with the generators of (a Poisson analog) of the generalised Racah algebra $R(n)$ \cite{Kuru2020, latini2020embedding}. Finally, we proceed towards the main goal of the work and show that the quadratic substructures mentioned above can in fact be interpreted as a classical (Poisson) realisation of some specific Racah subalgebras of rank one. The latter can be understood, at a fixed $n > 3$, as the result of the application of a family of $n-2$ injective morphisms of $R(3)$ into $R(n)$.  As a consequence, we show that each of the $n-2$  quadratic subalgebras is isomorphic to a Poisson analog of the rank one Racah algebra $R(3)$. We conclude the Section by providing the classical Casimir functions associated to each of the $n-2$ classical quadratic algebras.
	
	\item In Section \ref{sec4} we deal with the quantum case. The starting point, to construct the left and right quantum integrals, is the Lie algebra $\mathfrak{sl}(2, \mathbb{R})$ endowed with the (primitive) coassociative coproduct $\Delta$. In this case, the realisation is given in terms of linear differential operators. We rely on injective morphisms to establish the explicit commutation relations of these substructures and extend the results to the quantum case. In particular, we construct the quantum analog of the substructures and show that they give rise to a quantum realisation of the rank one embedded Racah algebras. Again, we conclude the Section by providing the quantum Casimir operators associated to each of the above-mentioned quadratic algebras. We mention that the classical case can be obtained from the quantum case by performing the appropriate $\hbar \to 0$ limit, so that all the formulas appearing in the paper can be easily compared.
	
	\item Section \ref{sec5} is devoted to the concluding remarks.
\end{itemize}

\section{Coalgebra symmetry approach to superintegrability: left and right Casimir invariants}

\label{sec2}

 \noindent We start by recalling some definitions following  \cite{Ballesteros1996, 0305-4470-31-16-009, 1742-6596-175-1-012004}.
\begin{definition}
A coalgebra $(\mathfrak{A},\Delta)$ is a (unital, associative) algebra $\mathfrak{A}$ endowed with a coproduct map:

$$\Delta: \mathfrak{A} \to \mathfrak{A}  \otimes \mathfrak{A}  $$

which is coassociative:

$$(\Delta \otimes \text{id}) \circ \Delta = (\text{id} \otimes \Delta) \circ \Delta \, $$

meaning that the following diagram is commutative:
\begin{center}
	\begin{tikzcd}
\mathfrak{A}\arrow[d,"\Delta"]\arrow[r,"\Delta"] 
&	\mathfrak{A} \otimes \mathfrak{A} \arrow[d, "\Delta\otimes \text{id}"]
\\\mathfrak{A} \otimes \mathfrak{A}  \arrow[r,"\text{id} \otimes \Delta"]
&
\mathfrak{A} \otimes \mathfrak{A}\otimes \mathfrak{A}
\end{tikzcd}
\end{center}
and it is an algebra homomorphism from $\mathfrak{A}$ to  $\mathfrak{A} \otimes \mathfrak{A}$:
$$\Delta (X \cdot  Y) = \Delta (X)  \cdot \Delta (Y) \qquad \forall \, X, Y \in \mathfrak{A} \, .$$
\end{definition}
\begin{definition}
	$(\mathfrak{A}, \Delta)$ is a Poisson coalgebra if $\mathfrak{A}$ is a Poisson algebra and $\Delta$ is a Poisson homomorphism between $\mathfrak{A}$ and $\mathfrak{A} \otimes \mathfrak{A}$, i.e.:
	$$\Delta\bigl(\{X, Y\}_\mathfrak{A}\bigl)=\{\Delta(X),\Delta(Y)\}_{\mathfrak{A} \otimes\mathfrak{A}} \quad \forall \, X,Y \in \mathfrak{A} \, ,$$
\end{definition}

\noindent with respect to the standard Poisson structure on $\mathfrak{A} \otimes \mathfrak{A}$:
$$\{X \otimes Y, W \otimes Z\}_{\mathfrak{A} \otimes \mathfrak{A}}:=\{X, W\}_{\mathfrak{A}} \otimes YZ+XW \otimes \{Y,Z\}_{\mathfrak{A}} \, , \quad X,Y,Z,W \in \mathfrak{A} \, .$$

 \noindent  Let us consider a Poisson coalgebra with $\ell$ generators $J_i$ ($i=1, \dots, \ell$) and a non-linear Casimir $C(J_1, \dots, J_\ell)$. Starting from the coproduct $\Delta=\Delta^{[2]}$ it is possible to define the $m$-th coproduct map, with $2<m\leq n$ in two possible ways. The first one is through the following recursive application of $\Delta^{[2]}$ (which is both a Poisson and an algebra homomorphism):
 \begin{equation}
\Delta_L^{[m]}:=(id \otimes  \dots ^{m-2)}\otimes id \otimes \Delta^{[2]}) \circ \Delta_L^{[m-1]} \, ,
\label{eq:leftCas}
\end{equation}
whereas the second one, that comes from a generalization of the coassociativity property to an arbitrary number of tensor product of $\mathfrak{A}$, by the following recursive one:
\begin{equation}
\Delta^{[m]}_R:=(\Delta^{[2]}\otimes id  \otimes  \dots ^{m-2)} \otimes id) \circ \Delta_R^{[m-1]}  \, .
\label{rightCas}
\end{equation}

\noindent  The main point now resides in the fact that, although the $n$-th coproduct of any generator of the algebra is the same due to the coassociativity of the $n$-th coproduct ($\Delta^{[n]}_L(J_i)=\Delta_R^{[n]}(J_i)$), lower dimensional ones with $ 2 < m < n$ will not be equivalent. This is because, by labelling with $1 \otimes 2 \otimes \dots \otimes n$ the sites of the chain $\mathfrak{A} \otimes \mathfrak{A}  \otimes \dots^{n)} \otimes \mathfrak{A}$, the \emph{left} and \emph{right} coproducts will be defined on the tensor product space $1 \otimes 2 \otimes \dots \otimes m$ and $(n-m+1) \otimes (n-m+2) \otimes \dots \otimes n$, respectively.
As a consequence, the so-called \emph{left} and \emph{right} Casimirs arise:
\begin{align}
&C^{[m]}:= \Delta_L^{[m]}(C(J_1, \dots, J_\ell))=C(\Delta_L^{[m]}(J_1), \dots, \Delta_L^{[m]}(J_\ell)) \label{eq:casimirs1}\\
&C_{[m]}:= \Delta^{[m]}_R(C(J_1, \dots, J_\ell))=C(\Delta^{[m]}_R(J_1), \dots, \Delta^{[m]}_R(J_\ell)) \, ,
\label{eq:casimirs2}
\end{align}
 with  $C^{[n]}=C_{{[n]}}$ because of the coassociativity property. 
\noindent Here the main result: if we consider a $n$-site Hamiltonian $H^{[n]}$ defined on the chain $\mathfrak{A} \otimes \mathfrak{A} \otimes  \dots^{n)} \otimes \mathfrak{A}$, i.e. the $n$-th coproduct of a (smooth/formal power series) function of the generators $J_i$ (taken as the initial 1-site Hamiltonian), i.e.:
\begin{equation}
H^{[n]}:=H(\Delta_L^{[n]}(J_1), \dots, \Delta_L^{[n]}(J_\ell))=H(\Delta^{[n]}_R(J_1), \dots, \Delta^{[n]}_R(J_\ell)) \, ,
\label{Hamcoal}
\end{equation}
\noindent then, for $m,m'=1, \dots, n$, the following Poisson commutation relations hold:
\begin{align}
&\{C^{[m]}, H^{[n]}\}_{\mathfrak{A} \otimes \mathfrak{A}  \dots ^{n)}\otimes  \mathfrak{A} }\,\,=0 \, , \\
&\{C^{{[m]}}, C^{[m']}\}_{\mathfrak{A} \otimes \mathfrak{A}  \dots ^{n)}\otimes  \mathfrak{A} } = 0 \, ,
\end{align}
where $C^{[1]} \equiv C$, so that the set of left Casimirs Poisson commute with the $n$-sites Hamiltonian $H^{[n]}$. Moreover, they define a set of involutive functions. The explicit proof can be found in \cite{0305-4470-31-16-009}. The same result can be rephrased for the set of right Casimirs. In view of this, the idea is to introduce explicit systems for which these algebraic results apply. In particular, the connection with classical Hamiltonian systems relies on the use of symplectic realisations. They allow to link the abstract algebraic construction to explicit classical mechanical models. In particular, once a symplectic realisation $D: \mathfrak{A} \to C^{\infty}(x,p)$ is given, the Hamiltonian $H^{[n]}$ will result in a function of the $2n$ canonical coordinates $(x_i,p_i)$ and, at least in principle, will be superintegrable by construction (since it Poisson commutes with the left and right integrals of motion arising under the chosen symplectic realisation). We refer the reader to \cite{Ballesteros_2008} for a comprehensive discussion on the integrability conditions for these models. Once the realisation $D$ is fixed, the bracket is the canonical one, i.e.:
\begin{equation}
\{f, g\}_{\mathfrak{A} \otimes \dots ^{n)}\otimes \mathfrak{A}} = \sum_{j=1}^n\bigl(\partial_{x_j}f\partial_{p_j}g-\partial_{p_j}f\partial_{x_j}g\bigl) \, \qquad f,g \in C^{\infty}(\bx, \bp).
\label{eq:canpoi}
\end{equation}
	\noindent Notice that the above algebraic construction also works for quantum mechanical systems, where commutators replace Poisson brackets. Of course, usual ordering problems have to be taken into account due to noncommutativity. Also, the construction can be extended when we deal with a Lie algebra possessing more than just one Casimir, say $r$ Casimirs, among which $R$ of them being non-linear ones (they are the ones of interest in relation to integrability). For all the details we refer the reader to  \cite{Ballesteros1996, 0305-4470-31-16-009, 1742-6596-175-1-012004}, where several applications (for deformed coalgebras also) can be found.
	\section{Quadratic substructures generated by the left and right Casimirs as rank one Racah subalgebras: The classical case}
	\label{sec3}
\noindent For our purposes, we are interested in this work to the Lie-Poisson (co)algebra $\mathfrak{A}=\mathfrak{sl}(2,\mathbb{R})\simeq \mathfrak{su}(1,1)$. The latter is endowed with $\ell=3$ generators $(J_{\pm}, J_3)$ and it is defined through the following Lie-Poisson brackets:
\begin{equation}
\{J_-,J_+\}= 2 J_3 \quad \{J_3, J_{\pm}\}=\pm  J_\pm \, .
\label{eq:LiePoisson}
\end{equation}
The coproduct $\Delta: \mathfrak{A} \to  \mathfrak{A} \otimes  \mathfrak{A} $ is primitive:
\begin{equation}
\Delta(J_\sigma)=J_\sigma \otimes 1 + 1 \otimes J_\sigma \, , \qquad \Delta(1)=1 \otimes 1 \, ,
\end{equation}
\noindent and the (non-linear) Casimir read:
\begin{equation}
C(J_+,J_-,J_3)=J_3^2-J_+J_- \, .
\label{eq:casimir}
\end{equation}

\noindent A one-dimensional symplectic realisation $D$ for the generators of this Lie-Poisson coalgebra is:
\begin{equation}
J_+^{[1]} := D(J_+)=\frac{1}{2}\bigl(p_1^2+\frac{a_1}{x_1^2} \bigl)\quad J_-^{[1]}: =D(J_-)= \frac{1}{2}x_1^2\quad J_3^{[1]} :=D(J_3)= \frac{1}{2}x_1 p_1 \, .
\label{1Drea}
\end{equation}
Thus, under the canonical Poisson brackets, we realise the initial Lie-Poisson algebra:
\begin{equation}
\{J^{[1]}_-,J^{[1]}_+\}= 2 J^{[1]}_3 \quad \{J^{[1]}_3, J^{[1]}_{\pm}\}=\pm  J^{[1]}_\pm \, .
\label{1DreaPoisn}
\end{equation}
The Casimir function, at a fixed realisation, reads:
\begin{equation}
C^{[1]}:=D(C(J_+,J_-,J_3))=-a_1/4 \, .
\label{eq:cas1rea}
\end{equation}
The $n$D realisation of the generators is constructed by applying the $n$-th coproduct $\Delta^{[n]}: \mathfrak{A} \to  \mathfrak{A} \otimes  \mathfrak{A} \otimes \dots ^{n)}\otimes  \mathfrak{A}$ to the basis generators $J_{\sigma}$ ($\sigma=\pm,3$) and by taking the explicit realisation, so that:
\begin{equation}
J_+^{[n]} = \frac{1}{2}\bigl(\bp^2+\sum_{j=1}^n\frac{a_j}{x_j^2}\bigl) \quad J_-^{[n]} = \frac{1}{2}\bx^2 \quad J_3^{[n]} = \frac{1}{2}\bx \cdot \bp \, ,
\label{eq:nDrea}
\end{equation}
where we assume the $a_j$ to be arbitrary real parameters. This means that these new $nD$ generators close:
\begin{equation*}
\{J_-^{[n]},J_+^{[n]}\}=2 J_3^{[n]} \quad \{J_3^{[n]}, J_{\pm}^{[n]}\}=\pm  J^{[n]}_\pm \, ,
\label{eq:classnD}
\end{equation*}
\noindent through the $n$D canonical Poisson bracket \eqref{eq:canpoi}. The total Casimir $C^{[n]}$ at a fixed realisation reads:
\begin{equation}
C^{[n]}=-\frac{1}{4}\biggl(\sum_{1 \leq i< j}^n \biggl(L_{ij}^2+a_i \frac{x_j^2}{x_i^2}+a_j \frac{x_i^2}{x_j^2}\biggl)+\sum_{i=1}^n a_i\biggl) \, ,
\end{equation}
with $L_{ij} = x_i p_j - x_j p_i$. This object can be seen just as the common left-right Casimir obtained when $m=n$. In fact, the left and right Casimirs \eqref{eq:casimirs1}-\eqref{eq:casimirs2} give rise, at a fixed realisation, to the following left and right classical integrals of motion:
\begin{align}
&C^{[m]}=-\frac{1}{4}\biggl(\hskip 0.4cm\sum_{1 \leq i< j}^m \hskip 0.4cm\biggl(L_{ij}^2+a_i \frac{x_j^2}{x_i^2}+a_j \frac{x_i^2}{x_j^2}\biggl)+\hskip 0.4cm\sum_{i=1}^m \hskip 0.4cm a_i\biggl)\\
&C_{[m]}=-\frac{1}{4}\biggl(\sum_{n-m+1 \leq i< j}^n \biggl(L_{ij}^2+a_i \frac{x_j^2}{x_i^2}+a_j \frac{x_i^2}{x_j^2}\biggl)+\sum_{i=n-m+1}^n a_i\biggl) \, ,
\label{eq:leftrigh}
\end{align}
where $C^{[1]}=-a_1/4$, $C_{[1]}=-a_n/4$ are just constants and $C^{[n]}=C_{[n]}$ as a consequence of the coassociativity property of the coproduct $\Delta^{[n]}$ on $\mathfrak{A} \otimes \mathfrak{A} \otimes \dots ^{n)}\otimes \mathfrak{A}$. Therefore, as a consequence of the general theory, the above $2n- 3$ functions Poisson commute with any arbitrary smooth function $H^{[n]}=H(J_+^{[n]}, J_-^{[n]}, J_3^{[n]})$ of the generators: 
\begin{equation}
\{C^{[m]}, H^{[n]}\}=\{C_{{[m]}}, H^{[n]}\}=0 \, .
\label{poisscomm}
\end{equation}
Thus, any Hamiltonian that is endowed with the above coalgebra symmetry (meaning that can be expressed as in \eqref{Hamcoal})  will be automatically endowed with the $2n-3$ conserved quantities coming from the left and right coproducts. Moreover, the left and right Casimirs define two sets composed by $n-1$ functions in involution ($m,m'=1, \dots, n$):
\begin{equation}
\{C^{[m]}, C^{[m']}\}=\{C_{{[m]}}, C_{{[m']}}\}=0 \, .
\label{involution}
\end{equation}
If we add the Hamiltonian to the above sets, they  will be both composed by $n$ functions in involution. 

To summarise, the construction based on coalgebras led us to define two subsets, say $Y_L$ and $Y_R$, composed by $n-1$ left and right constants of motion in involution, i.e.:

\begin{equation}
Y_L=\{C^{[1]}, \dots, C^{[n]}\} \, , \quad Y_R=\{C_{[1]}, \dots, C_{[n]}\} \, ,
\label{eq:ylyr}
\end{equation}

\noindent the total Casimir $C^{[n]}=C_{[n]}$ being a common element of the sets and $C^{[1]}$, $C_{[1]}$ just constants. 
Once arrived at this point, let us observe that the total Casimir can be rewritten in the following way:

\begin{align}
C^{[n]}&=-\frac{1}{4}\biggl(\sum_{1 \leq i< j}^n \biggl(L_{ij}^2+a_i \frac{x_j^2}{x_i^2}+a_j \frac{x_i^2}{x_j^2}\biggl)+\sum_{i=1}^n a_i\biggl) \nonumber \\
&=-\frac{1}{4}\biggl(\sum_{1 \leq i< j}^n L_{ij}^2+\boldsymbol{x}^2 \sum_{j=1}^n \frac{a_j}{x_j^2}\biggl) \nonumber \\
&=\sum_{1 \leq i< j}^n C_{ij}-(n-2)\sum_{i=1}^nC_{i} \, ,
\label{eq:totcas}
\end{align}

\noindent where we defined the quantities:
\begin{equation}
C_{ij}:=-\frac{1}{4}\biggl( L_{ij}^2+a_i \frac{x_j^2}{x_i^2}+a_j \frac{x_i^2}{x_j^2} +a_i+a_j\biggl) \qquad C_{i}:= -\frac{a_i}{4} \, .
\label{eq:racahgen}
\end{equation}

\noindent Also,  the left and right Casimirs can be re-expressed as:
\begin{equation}
C^{[m]}=\sum_{1 \leq i< j}^m C_{ij}-(m-2)\sum_{i=1}^m C_i \, ,\qquad C_{[m]}=\sum_{n-m+1 \leq i< j}^n C_{ij}-(m-2)\sum_{i=n-m+1}^n C_i \, .
\label{lrRac}
\end{equation}
\noindent These new expressions for the left and right partial Casimirs emphasize, in a more general perspective, their connection with the (abstract) generators of the rank $n-2$ generalised Racah algebra $R(n)$ \cite{bie2020racah, Gaboriaud_2019}, the latter generated by the two indices generators $C_{ij}$ together with the one index generators $C_i$, here appearing as the building blocks of the left and right Casimirs in the given realisation of the Poisson-Lie algebra. Notice that the \textquotedblleft central elements\textquotedblright $C_i$ are realised here just as the constants appearing in \eqref{eq:racahgen}, which underline the one-particle symplectic realisation used on the generic $i$-th site of the chain $\mathfrak{A} \otimes \mathfrak{A} \otimes  \dots ^{n)}\otimes \mathfrak{A}$.
\noindent At this point, we can go further by introducing the quantities:
\begin{equation}
P_{ij}:=C_{ij}-C_i-C_j \, ,
\label{eq:pij}
\end{equation}
\noindent in such a way to simplify the above expressions to the following ones:
\begin{equation}
C^{[m]}=\sum_{1 \leq i< j}^m P_{ij}+\sum_{i=1}^m C_i \, , \qquad C_{[m]}=\sum_{n-m+1 \leq i< j}^n P_{ij}+\sum_{i=n-m+1}^n C_i \, .
\label{newcas}
\end{equation}
The defining relations from $P_{ij}$ are those of the classical analog of the Racah algebra $R(n)$ \cite{Kuru2020, latini2020embedding}:
\begin{align}
&\{P_{ij},P_{jk}\}=:2 F_{ijk} \hskip 5.2cm (F_{ijk}=-{F_{jik}}=-F_{ikj}) \label{eq:PoissonRacah1}\\
&\{P_{jk}, F_{ijk}\}=P_{ik}P_{jk}-P_{jk}P_{ij}+2P_{ik}C_j-2P_{ij}C_k\\
&\{P_{kl}, F_{ijk}\}=P_{ik}P_{jl}-P_{il}P_{jk}\\
&\{F_{ijk}, F_{jkl}\}=F_{jkl}P_{ij}-F_{ikl}(P_{jk}+2C_j)-F_{ijk}P_{jl}\\
&\{F_{ijk}, F_{klm}\}=F_{ilm}P_{jk}-P_{ik}F_{jlm} \, .
\label{eq:PoissonRacah5}
\end{align}

\noindent Finally, if we introduce the new definitions $\mathcal{L}_k := C^{[k]}$ and $\mathcal{R}_k := C_{[n-k+1]}$, for $k=1, \dots, n$, we get:

\begin{equation}
\mathcal{L}_k :=\sum_{1 \leq i< j}^k P_{ij}+\sum_{i=1}^k C_i  \, ,\qquad \mathcal{R}_k := \sum_{k \leq i< j}^n P_{ij}+\sum_{i=k}^n C_i   \, .
\label{lrpij}
\end{equation}
\noindent We thus have, at a fixed $n$ for $k=1, \dots, n$, the following identifications:

\vskip 0.5cm

\begin{enumerate}
\item[\label=$-$]  $(n;\{k\})=(1;\{1\})$
\begin{equation}
\mathcal{L}_1=C^{[1]}=C_1=C_{[1]}= \mathcal{R}_1 \, .
\label{id}
\end{equation}

\item[\label=$-$]  $(n;\{k\})=(2;\{1,2\})$
\begin{equation}
\mathcal{L}_1=C^{[1]} \, , \quad  \mathcal{L}_2=C^{[2]}=C_{[2]}=\mathcal{R}_1 \, ,\quad \mathcal{R}_2 = C_{[1]} \, ,
\label{id2}
\end{equation}
with:
\begin{align*}
&\mathcal{L}_1=C_1\\
& \mathcal{L}_2=P_{12}+C_1+C_2=C_{12}\\
&\mathcal{R}_1=P_{12}+C_1+C_2=C_{12}\\
&\mathcal{R}_2 = C_2 \, .
\end{align*}
\item[\label=$-$]  $(n;\{k\})=(3;\{1,2,3\})$
\begin{equation}
\mathcal{L}_1=C^{[1]} \, ,\quad  \mathcal{L}_2=C^{[2]} \, , \quad  \mathcal{L}_3=C^{[3]}=C_{[3]}=\mathcal{R}_1\, ,\quad \mathcal{R}_2 = C_{[2]} \, , \quad \mathcal{R}_3 = C_{[1]} \, ,
\label{id3}
\end{equation}
with:
\begin{align*}
&\mathcal{L}_1=C_1 \\
& \mathcal{L}_2=P_{12}+C_1+C_2=C_{12} \\
& \mathcal{L}_3=P_{12}+P_{13}+P_{23}+C_1+C_2+C_3=C_{12}+C_{13}+C_{23}-C_1-C_2-C_3=C_{123}\\
&\mathcal{R}_1=P_{12}+P_{13}+P_{23}+C_1+C_2+C_3=C_{12}+C_{13}+C_{23}-C_1-C_2-C_3=C_{123}\\
&\mathcal{R}_2 = P_{23}+C_2+C_3=C_{23} \\
& \mathcal{R}_3 = C_3  \, .
\end{align*}
\vskip 0.25cm
\hskip 0.75cm \vdots
\vskip 0.5cm
\item[\label=$-$] $(n;\{k\})=(n;\{1, \dots, n\})$
\begin{align}
&\mathcal{L}_1=C^{[1]} \, ,\quad \mathcal{L}_2=C^{[2]} \, ,\quad \,\quad \mathcal{L}_3=C^{[3]}  \, , \qquad \dots \quad , \quad   \mathcal{L}_n=C^{[n]}\\
&\mathcal{R}_1=C_{[n]} \, , \quad  \mathcal{R}_2 = C_{[n-1]} \, , \quad \mathcal{R}_3 = C_{[n-2]} \, , \quad \dots \quad \, ,\quad   \mathcal{R}_n = C_{[1]}   \, ,
\label{idn}
\end{align}{}
with:
\begin{align*}
&\mathcal{L}_1=C_1 \\
& \mathcal{L}_2=P_{12}+C_1+C_2=C_{12} \\
& \mathcal{L}_3=P_{12}+P_{13}+P_{23}+C_1+C_2+C_3=C_{12}+C_{13}+C_{23}-C_1-C_2-C_3=C_{123}\\
& \vdots \\
&\mathcal{L}_{n-1}=\sum_{1 \leq i< j}^{n-1} P_{ij}+\sum_{i=1}^{n-1} C_i=\sum_{1 \leq i< j}^{n-1} C_{ij}-(n-3)\sum_{i=1}^{n-1}C_{i}=C_{1\dots n-1}\\
&\mathcal{L}_n=\mathcal{R}_1=\sum_{1 \leq i< j}^n P_{ij}+\sum_{i=1}^n C_i=\sum_{1 \leq i< j}^n C_{ij}-(n-2)\sum_{i=1}^nC_{i}=C_{1\dots n}\\
&\mathcal{R}_2 = \sum_{2 \leq i< j}^n P_{ij}+\sum_{i=2}^n C_i  =\sum_{2 \leq i< j}^n C_{ij}-(n-3)\sum_{i=2}^nC_{i}=C_{2\dots n}\\
& \mathcal{R}_3 = \sum_{3 \leq i< j}^n P_{ij}+\sum_{i=3}^n C_i  =\sum_{3 \leq i< j}^n C_{ij}-(n-4)\sum_{i=3}^nC_{i}=C_{3 \dots n} \\
& \vdots \\
& \mathcal{R}_{n-1} = P_{n-1\,n}+C_{n-1}+C_n=C_{n-1 \,n}  \\
&\mathcal{R}_n = C_n \, .
\end{align*}
\end{enumerate}


\noindent The left and right Casimirs \eqref{lrpij} define the generators of the universal quadratic substructures appearing in \cite{Latini_2019, 10.1088/1751-8121/aac111}. 

Our main goal now   is to show that the above-mentioned substructures can be reinterpreted, at a fixed $n > 3$, as the images of $n-2$ injective morphisms of $R(3)$ into $R(n)$ and that, as a direct consequence, they turn out to be isomorphic to the rank one Racah algebra $R(3)$.
As a first step we rephrase this in the classical framework and show that for $n=3$ the only substructure that arises from the left and right Casimir invariants is a Poisson analog of the Racah algebra $R(3)$.

\subsection{The $\boldsymbol{n=3}$ case: Poisson analog of the rank one Racah algebra $\boldsymbol{R(3)}$}
\label{sub3.1}
\noindent The generators of the Poisson analog of the rank one Racah algebra are: $\{C_1, C_2, C_3, C_{12}, C_{13}, C_{23}, C_{123} \}$ with:
\begin{equation}
C_{123}=C_{12}+C_{13}+C_{23}-C_1-C_2-C_3  \, ,
\label{lineareq}
\end{equation}
where $C_1$, $C_2$, $C_3$ and $C_{123}$ are central elements.
\noindent Now, by defining:
\begin{equation}
\mathcal{F}:=\frac{1}{2}\{C_{12},C_{23}\}=\frac{1}{2}\{C_{23},C_{13}\}=\frac{1}{2}\{C_{13},C_{12}\}
\label{Fijk}
\end{equation}

\noindent we got the following defining Poisson brackets:
\begin{align}
&\{C_{12},\mathcal{F}\}=C_{23}C_{12}-C_{12}C_{13}+(C_2-C_1)(C_3-C_{123}) \label{1rac}\\
&\{C_{23},\mathcal{F}\}=C_{13}C_{23}-C_{23}C_{12}+(C_3-C_2)(C_1-C_{123})  \label{2rac}\\
&\{C_{13},\mathcal{F}\}=C_{12}C_{13}-C_{13}C_{23}+(C_1-C_3)(C_2-C_{123}) \, .  \label{3rac}
\end{align}
From \eqref{lineareq}, using the properties of the Poisson brackets, we obtain the relations:
\begin{equation}
\{C_{23}, C_{12}+C_{13}\}=0 \qquad \{C_{12}, C_{13}+C_{23}\}=0 \qquad \{C_{13}, C_{12}+C_{23}\}=0 \, .
\label{eq:DK}
\end{equation}
Also, the following sum is zero:
\begin{equation}
\{C_{12},\mathcal{F}\}+\{C_{23},\mathcal{F}\}+\{C_{13},\mathcal{F}\}=0 \, .
\label{sum}
\end{equation}
\noindent Now, if we use \eqref{lineareq} to explicitate the $C_{13}$ generator, i.e.:
\begin{equation}
C_{13}=C_{123}-C_{12}-C_{23}+C_1+C_2+C_3 \,
\label{c13}
\end{equation}
and we substitute this expressions into  \eqref{1rac}-\eqref{2rac}, taking into account \eqref{sum}, we get:
\begin{align}
&\{C_{12},\mathcal{F}\}=+C_{12}^2+2C_{23}C_{12}-(C_{123}+C_1+C_2+C_3)C_{12}+(C_2-C_1)(C_3-C_{123}) \label{eq:sub321}\\
&\{C_{23},\mathcal{F}\}=-C_{23}^2-2C_{12}C_{23}+(C_{123}+C_1+C_2+C_3)C_{23}+(C_3-C_2)(C_1-C_{123}) \label{eq:sub322}\, .
\end{align}

\noindent We recognize this quadratic Poisson algebra to be nothing but the substructure $(n,{k})=(3,2)$, which we indicate as\footnote{From now on we use the notation $S_{(n,k)}$ to indicate the generic $k$-th substructure at fixed $n$.} $S_{(3,2)}$ as reported in \cite{Latini_2019}.  Notice that in this presentation the operator $C_{13}$, being a linear combination of the others, does not appear explicitly. Now, if we consider the relations among the generators of the Racah algebra and the left and right Casimirs in dimension $n=3$, as we reported previously, this connection appears more evident. In fact, if substitute the expressions we get:
\begin{align}
&\{\mathcal{L}_2,\mathcal{F}\}=+\mathcal{L}_2^2+2\mathcal{R}_2\mathcal{L}_2-(\mathcal{L}_3+\mathcal{L}_1+C_2+\mathcal{R}_3)\mathcal{L}_2+(C_2-\mathcal{L}_1)(\mathcal{R}_3-\mathcal{L}_3) \label{s32c1}\\
&\{\mathcal{R}_2,\mathcal{F}\}=-\mathcal{R}_2^2-2\mathcal{L}_2\mathcal{R}_2+(\mathcal{L}_3+\mathcal{L}_1+C_2+\mathcal{R}_3)\mathcal{R}_2+(\mathcal{R}_3-C_2)(\mathcal{L}_1-\mathcal{L}_3)  \label{s32c2} \, ,
\end{align}
with:
\begin{equation}
\mathcal{F}:=\frac{1}{2}\{\mathcal{L}_2,\mathcal{R}_2\}=\frac{1}{2}\{\mathcal{R}_2,\mathcal{M}_2\}=\frac{1}{2}\{\mathcal{M}_2,\mathcal{L}_2\}\, ,
\label{FLRijk}
\end{equation}
where we introduced the new definition $\mathcal{M}_2:=\mathcal{L}_3-\mathcal{L}_2-\mathcal{R}_2+\mathcal{L}_1+C_2+\mathcal{R}_3$ for the generator $C_{13}$, the latter will be useful later in the generalization to higher dimensions.  Thus, on the basis of these results, we can provide an equivalent presentation for the substructure $S_{(3,2)}$, which reads:
\begin{align}
&\{\mathcal{L}_2,\mathcal{F}\}\,=\mathcal{R}_2\mathcal{L}_2-\mathcal{L}_2\mathcal{M}_2+(C_2-\mathcal{L}_1)(\mathcal{R}_3-\mathcal{L}_3) \label{1sb}\\
&\{\mathcal{R}_2,\mathcal{F}\}\,=\mathcal{M}_2\mathcal{R}_2-\mathcal{R}_2\mathcal{L}_2+(\mathcal{R}_3-C_2)(\mathcal{L}_1-\mathcal{L}_3)  \label{2sb}\\
&\{\mathcal{M}_2,\mathcal{F}\}=\mathcal{L}_2\mathcal{M}_2-\mathcal{M}_2\mathcal{R}_2+(\mathcal{L}_1-\mathcal{R}_3)(C_2-\mathcal{L}_3) \, .  \label{3sb}
\end{align}
In conclusion, we observed that the substructure arising in dimension $n=3$ in \cite{Latini_2019}  represents a classical realisation (given in terms of Poisson brackets) of the rank one Racah algebra $R(3)$, the latter being generated by the set $\{\mathcal{L}_1, C_2, \mathcal{R}_3, \mathcal{L}_{2}, \mathcal{M}_{2}, \mathcal{R}_{2}, \mathcal{L}_{3}=\mathcal{R}_1 \}$, where $C_2=-a_2/4$ in the given realisation and $\mathcal{M}_{2}$ appears as a linear combination of the other generators.

\subsection{Higher dimensional classical substructures, injective morphisms and (Poisson) Racah algebras $\boldsymbol{R^{K_1,K_2,K_3}(3)}$}
\label{sec3.2}
\noindent The main aim of this section is to provide a characterisation of the  substructures $S_{(n,k)}$, at a fixed $n>3$ for $k=2, \dots, n-1$ with the help of the following \cite{De_Bie_2017, bie2020racah, Crampe2020, DeBie}:

\begin{lemma}
	Let $\{K_1, K_2, K_3\}$ be a set composed by three disjoint subsets of the set $[n]:=\{1,\dots, n\}$. Define $K_B:=\cup_{q \in B} K_q$ with $B \subset [3]$. Then, the following map:
\begin{align*}
	\theta: \,&R(3) \to R(n)\\
	& \,C_B \,\,\, \to C_{K_B}
	\end{align*}
	is an injective morphism. Its image is denoted by $R^{K_1,K_2,K_3}(3)$ and it is isomorphic to the rank one Racah algebra.
	\label{Lemma}
\end{lemma}

\noindent This result appears to be a specific case of the more general one appearing in \cite{bie2020racah, DeBie}, where the authors deal with a general number $\textsf{k}$ of disjoint subsets $\{K_p\}_{p=1, \dots, \textsf{k}}$ of $[n]$. We refer to  \cite{De_Bie_2017} for the proof of this Lemma (see also \cite{Crampe2020} for the generalisation to the $\textsf{k}$-fold tensor product). So, the idea is to make use of the above result paraphrasing it in this classical context.
In particular, if we consider the generators of the Racah algebra $R(3)$, i.e. the set $\{C_1, C_2, C_3, C_{12}, C_{13}, C_{23}, C_{123}\}$, then their images under the map $\theta$: 
\begin{equation}
\small \{\theta(C_1),\theta(C_2), \theta(C_3), \theta(C_{12}), \theta(C_{13}), \theta(C_{23}), \theta(C_{123})\} = \{C_{K_1}, C_{K_2}, C_{K_3}, C_{K_1 \cup K_2}, C_{K_1 \cup K_3}, C_{K_2 \cup K_3}, C_{K_1 \cup K_2 \cup K_3} \} \, ,
\label{map}
\end{equation}
define a new set composed by seven generators that closes an algebra isomorphic to the rank one Racah algebra $R(3)$. In our classical setting, in terms of Poisson brackets, the algebra $R^{K_1,K_2,K_3}(3)$ is defined as:
\begin{align}
&\{C_{K_1 K_2},\mathcal{F}\}=C_{K_2K_3}C_{K_1 K_2}-C_{K_1 K_2}C_{K_1 K_3}+(C_{K_2}-C_{K_1})(C_{K_3}-C_{K_1 K_2 K_3}) \label{1rac_theta}\\
&\{C_{K_2 K_3},\mathcal{F}\}=C_{K_1K_3}C_{K_2K_3}-C_{K_2K_3}C_{K_1K_2}+(C_{K_3}-C_{K_2})(C_{K_1}-C_{K_1 K_2 K_3})  \label{2rac_theta}\\
&\{C_{K_1 K_3},\mathcal{F}\}=C_{K_1 K_2}C_{K_1 K_3}-C_{K_1 K_3}C_{K_2 K_3}+(C_{K_1}-C_{K_3})(C_{K_2}-C_{K_1 K_2 K_3}) \, , \label{3rac_theta}
\end{align}
with:
\begin{equation}
\mathcal{F}:=\frac{1}{2}\{C_{K_1 K_2},C_{K_2 K_3}\}=\frac{1}{2}\{C_{K_2 K_3},C_{K_1 K_3}\}=\frac{1}{2}\{C_{K_1 K_3},C_{K_1 K_2}\} \, .
\label{Fijk_theta}
\end{equation}
Here, we used the shorthand notation $C_{K_i K_j}:=C_{K_i \cup K_j}$ for $1 \leq i <j \leq 3$ and $C_{K_1 K_2 K_3}:= C_{K_1 \cup K_2 \cup K_3}$. Notice that the linear relation \eqref{lineareq} is rephrased here as:
\begin{equation}
C_{K_1 K_2 K_3}=C_{K_1 K_2}+ C_{K_1 K_3}+ C_{K_2 K_3}-C_{K_1}-C_{K_2}-C_{K_3} \, .
\label{eq}
\end{equation}
What we are interested to show in the following is that the substructures $S_{(n,k)}$ studied in \cite{Latini_2019, 10.1088/1751-8121/aac111}, at a fixed $n > 3$, can be interpreted as a realisation of the algebras arising as the images of $n-2$ injective morphisms: 
$$\theta_k: R(3) \to R(n) \quad  (k=2, \dots n-1) \, ,$$
for specific choices of the three disjoint subset $\{K_i\}_{i=1,2,3}$ as defined above. Thus, each substructure turns out to be isomorphic to a classical realisation of the rank one Racah algebra $R(3)$. Let us present this result as the following:

\begin{prop}
	\label{prop1}
	 Let $[n]:=\{1, \dots, n\}$ be a set with $n>3$ fixed. For each $k=2, \dots, n-1$ define $K_1^{(k)}:=\{1,\dots, k-1\}$, $K_2^{(k):}=\{k\}$ and $K_3^{(k)}:=\{k+1, \dots, n\}$ to be three disjoint subsets of the set $[n]$. Then, the images of the $n-2$ injective morphisms $\theta_k:$  $R(3) \to R(n)$ $(k=2, \dots, n-1)$:
\begin{align*}
&\theta_k(C_1)=C_{K_1^{(k)}} \equiv C_{1 \dots k-1}=\mathcal{L}_{k-1}\\
&\theta_k(C_2)=C_{K_2^{(k)}} \equiv C_k \\
&\theta_k(C_3)=C_{K_3^{(k)}} \equiv C_{k+1 \dots n}=\mathcal{R}_{k+1}\\
&\theta_k(C_{12})=C_{K_1^{(k)} K_2^{(k)}} \equiv C_{1 \dots k}=\mathcal{L}_k\\
&\theta_k(C_{23})=C_{K_2^{(k)} K_3^{(k)}} \equiv C_{k \dots n}=\mathcal{R}_k\\
&\theta_k(C_{13})=C_{K_1^{(k)} K_3^{(k)}} \equiv C_{1 \dots k-1 k+1 \dots n}=\mathcal{M}_k\\
&\theta_k(C_{123})=C_{K_1^{(k)} K_2^{(k)} K_3^{(k)}} \equiv C_{1\dots n}=\mathcal{L}_n=\mathcal{R}_1
\end{align*}
\noindent with: $$\mathcal{M}_{k}:=\mathcal{L}_n-\mathcal{L}_k-\mathcal{R}_k+\mathcal{L}_{k-1}+C_k+\mathcal{R}_{k+1} \, ,$$
result in the universal classical substructures $S_{(n,k)}$. Moreover, as a direct consequence of Lemma \ref{Lemma} they are isomorphic to the Poisson analog of the Racah algebra $R(3)$.  
\end{prop}
\begin{proof}
At a fixed $n>3$, for each $k=2, \dots, n-1$, we can define the algebras $R^{K_1^{(k)},K_2^{(k)},K_3^{(k)}}(3)$ as the images of the $n-2$ morphisms defined above and lift the relations valid for $R(3)$ into $R(n)$ (we can think the procedure as $n-2$ applications of Lemma \ref{Lemma}). This implies that each of the $n-2$  subalgebras are generated by the elements:
\begin{equation}
\{C_{K_1^{(k)}},C_{K_2^{(k)}},C_{K_3^{(k)}}, C_{K_1^{(k)} K_2^{(k)}}, C_{K_1^{(k)} K_3^{(k)}}, C_{K_1^{(k)} K_2^{(k)} K_3^{(k)}} \} 
\label{eq:elements}
\end{equation}
$C_{K_1^{(k)}},C_{K_2^{(k)}},C_{K_3^{(k)}}$ and $C_{K_1^{(k)} K_2^{(k)} K_3^{(k)}}$ playing the role of central elements for each of the $n-2$ substructures. The linear relation \eqref{lineareq} is lifted to the following $n-2$ linear relations (one relation for each $k=2, \dots,n-1$):
\begin{equation}
C_{K_1^{(k)} K_2^{(k)} K_3^{(k)}}=C_{K_1^{(k)} K_2^{(k)}}+ C_{K_1^{(k)} K_3^{(k)}}+ C_{K_2^{(k)} K_3^{(k)}}-C_{K_1^{(k)}}-C_{K_2^{(k)}}-C_{K_3^{(k)}} \, .
\label{eqkk}
\end{equation}
So, we have:
\begin{equation}
\mathcal{F}_k:=\frac{1}{2}\{C_{K_1^{(k)} K_2^{(k)} },C_{K_2^{(k)}  K_3^{(k)} }\}=\frac{1}{2}\{C_{K_2^{(k)}  K_3^{(k)} },C_{K_1^{(k)}  K_3^{(k)} }\}=\frac{1}{2}\{C_{K_1^{(k)}  K_3^{(k)} },C_{K_1^{(k)}  K_2^{(k)} }\} \, ,
\label{Fijk_thetak}
\end{equation}

\noindent together with:
\begin{align}
&\{C_{K_1^{(k)} K_2^{(k)}},\mathcal{F}_k\}=C_{K_2^{(k)}K_3^{(k)}}C_{K_1^{(k)} K_2^{(k)}}-C_{K_1^{(k)} K_2^{(k)}}C_{K_1^{(k)} K_3^{(k)}}+(C_{K_2^{(k)}}-C_{K_1^{(k)}})(C_{K_3^{(k)}}-C_{K_1^{(k)} K_2^{(k)} K_3^{(k)}}) \label{1rac_thetak}\\
&\{C_{K_2^{(k)} K_3^{(k)}},\mathcal{F}_k\}=C_{K_1^{(k)}K_3^{(k)}}C_{K_2^{(k)}K_3^{(k)}}-C_{K_2^{(k)}K_3^{(k)}}C_{K_1^{(k)}K_2^{(k)}}+(C_{K_3^{(k)}}-C_{K_2^{(k)}})(C_{K_1^{(k)}}-C_{K_1^{(k)} K_2^{(k)} K_3^{(k)}})  \label{2rac_thetak}\\
&\{C_{K_1^{(k)} K_3^{(k)}},\mathcal{F}_k\}=C_{K_1^{(k)} K_2^{(k)}}C_{K_1^{(k)} K_3^{(k)}}-C_{K_1^{(k)} K_3^{(k)}}C_{K_2^{(k)} K_3^{(k)}}+(C_{K_1^{(k)}}-C_{K_3^{(k)}})(C_{K_2^{(k)}}-C_{K_1^{(k)} K_2^{(k)} K_3^{(k)}}) \, . \label{3rac_thetak}
\end{align}
\noindent The following relations, the equivalent of \eqref{eq:DK}, also holds:
\begin{align}
&\{C_{K_2^{(k)}K_3^{(k)}},  C_{K_1^{(k)}K_2^{(k)}}+C_{K_1^{(k)}K_3^{(k)}}\}=0\\  &\{C_{K_1^{(k)}K_2^{(k)}}, C_{K_1^{(k)}K_3^{(k)}}+C_{K_2^{(k)}K_3^{(k)}}\}=0 \\ &\{C_{K_1^{(k)}K_3^{(k)}}, C_{K_1^{(k)}K_2^{(k)}}+C_{K_2^{(k)}K_3^{(k)}}\}=0 \, ,
\label{eq:DKk}
\end{align}
and the following sum is still zero when we lift it into $R(n)$:
\begin{equation}
\{C_{K_1^{(k)}K_2^{(k)}},\mathcal{F}_k\}+\{C_{K_2^{(k)}K_3^{(k)}},\mathcal{F}_k\}+\{C_{K_1^{(k)}K_3^{(k)}},\mathcal{F}_k\}=0 \, .
\label{sumk}
\end{equation}
What we need to show is the equivalence with the substructures $S_{(n,k)}$. To do this, we can proceed as we have done for the $n=3$ case in Subsection \ref{sub3.1}, as we are dealing with $n-2$ algebras isomorphic to $R(3)$.

For each $k=2, \dots, n-1$, let $K_1^{(k)}:=\{1,\dots, k-1\}$, $K_2^{(k):}=\{k\}$ and $K_3^{(k)}:=\{k+1, \dots, n\}$ be three specific disjoint subsets of $[n]$. With this choice, the generators of $R^{K_1^{(k)}, K_2^{(k)}, K_3^{(k)}}(3)$  are:
\begin{equation}
 \{\mathcal{L}_{k-1}, C_k, \mathcal{R}_{k+1}, \mathcal{L}_k, \mathcal{R}_k, \mathcal{M}_k, \mathcal{L}_n=\mathcal{R}_1\},
 \label{eq:gene}
 \end{equation}
 $\mathcal{L}_{k-1}$, $C_k$, $\mathcal{R}_{k+1}$ and $\mathcal{L}_n=\mathcal{R}_1$ being central elements associated to the $k$-th substructure. The linear relation \eqref{lineareq} is lifted to the $n-2$ linear relations:
\begin{equation}
\mathcal{L}_n = \mathcal{L}_{k}+\mathcal{R}_{k}+\mathcal{M}_{k}-\mathcal{L}_{k-1}-C_k-\mathcal{R}_{k+1}=\mathcal{R}_1 \, .
\label{lineqk}
\end{equation}
and the Poisson algebras $R^{K_1^{(k)},K_2^{(k)},K_3^{(k)}}(3)$ turn out to be:
\begin{align}
&\{\mathcal{L}_k,\mathcal{F}_k\}=\mathcal{R}_k\mathcal{L}_k-\mathcal{L}_k\mathcal{M}_k+(C_k-\mathcal{L}_{k-1})(\mathcal{R}_{k+1}-\mathcal{L}_n) \label{rac1}\\
&\{\mathcal{R}_k,\mathcal{F}_k\}=\mathcal{M}_k\mathcal{R}_k\,-\mathcal{R}_k\mathcal{L}_k\,+(\mathcal{R}_{k+1}-C_k)(\mathcal{L}_{k-1}-\mathcal{L}_n) \label{rac2}\\
&\{\mathcal{M}_k,\mathcal{F}_k\}=\mathcal{L}_k\mathcal{M}_k-\mathcal{M}_k\mathcal{R}_k+(\mathcal{L}_{k-1}-\mathcal{R}_{k+1})(C_k-\mathcal{L}_{n}) \label{rac3}\, ,
\end{align}
where:
\begin{equation}
\mathcal{F}_k:=\frac{1}{2}\{\mathcal{L}_k,\mathcal{R}_k\}=\frac{1}{2}\{\mathcal{R}_k,\mathcal{M}_{k}\}=\frac{1}{2}\{\mathcal{M}_{k},\mathcal{L}_k\} \, .
\label{funk}
\end{equation}
The relations given in \eqref{eq:DK}, are lifted to the following ones:
\begin{equation}
\{\mathcal{R}_k, \mathcal{L}_{k}+\mathcal{M}_{k}\}=0 \, ,\qquad \{\mathcal{L}_k, \mathcal{M}_{k}+\mathcal{R}_{k}\}=0  \, ,\qquad \{\mathcal{M}_{k}, \mathcal{L}_{k}+\mathcal{R}_{k}\}=0 \, .
\label{DKk}
\end{equation}
Moreover, relation \eqref{sum}, is now lifted to:
\begin{equation}
\{\mathcal{L}_k,\mathcal{F}_k\}+\{\mathcal{R}_k,\mathcal{F}_k\}+\{\mathcal{M}_k,\mathcal{F}_k\}=0 \,.
\label{sumkk}
\end{equation}
 In total analogy to the three dimensional case, we use \eqref{lineqk} to explicitate the $\mathcal{M}_k$ generators, i.e.:
\begin{equation}
\mathcal{M}_{k}=\mathcal{L}_n-\mathcal{L}_k-\mathcal{R}_k+\mathcal{L}_{k-1}+C_k+\mathcal{R}_{k+1} \, ,
\label{Mk}
\end{equation}
and we substitute these expressions into  \eqref{rac1}-\eqref{rac2}, taking into account \eqref{sumkk}. In this way, we get:
\begin{align}
&\{\mathcal{L}_k,\mathcal{F}_k\}=+\mathcal{L}_k^2+2\mathcal{R}_k \mathcal{L}_k-(\mathcal{L}_n+\mathcal{L}_{k-1}+C_k+\mathcal{R}_{k+1})\mathcal{L}_k+(C_k-\mathcal{L}_{k-1})(\mathcal{R}_{k+1}-\mathcal{L}_n) \label{csubstructure1}\\
&\{\mathcal{R}_k,\mathcal{F}_k\}=-\mathcal{R}_k^2-2\mathcal{L}_k \mathcal{R}_k+(\mathcal{L}_n+\mathcal{L}_{k-1}+C_k+\mathcal{R}_{k+1})\mathcal{R}_k+(\mathcal{R}_{k+1}-C_{k})(\mathcal{L}_{k-1}-\mathcal{L}_n) \label{csubstructure2} \, ,
\end{align}

\noindent where $\mathcal{L}_n=C^{[n]}=C_{[n]}=\mathcal{R}_1$. At a fixed $n$, for each $k=2, \dots, n-1$ these Poisson commutation relations define the universal quadratic substructures $S_{(n,k)}$ generated by the classical left and right Casimir invariants that can be found in \cite{Latini_2019}. \newline
\end{proof}
\begin{rem}
If we let $n$ to be equal three in the construction, so that $k=2$, we obtain the three disjoint subsets $K_1^{(2)}=\{1\}$, $K_2^{(2)}=\{2\}$ and $K_3^{(2)} =\{3\}$. In this case, the morphism $\theta_2 : R(3) \to R(3)$ acts as follows:
\begin{align*}
	&\theta_2(C_1)=C_{K_1^{(2)}} \equiv C_1=\mathcal{L}_1\\
	&\theta_2(C_2) = C_{K_2^{(2)}}\equiv C_2 \\
	&\theta_2(C_3) = C_{K_3^{(2)}} \equiv C_3=\mathcal{R}_3\\
	&\theta_2(C_{12})=C_{K_1^{(2)} K_2^{(2)}} \equiv C_{12}=\mathcal{L}_2\\
	&\theta_2(C_{23})=C_{K_2^{(2)} K_3^{(2)}} \equiv C_{23}=\mathcal{R}_2\\
	&\theta_2(C_{13})=C_{K_1^{(2)} K_3^{(2)}} \equiv C_{13}=\mathcal{M}_2\\
	&\theta_2(C_{123})=C_{K_1^{(2)} K_2^{(2)}K_3^{(2)}} \equiv C_{123}=\mathcal{L}_3=\mathcal{R}_1 \, ,
\end{align*}
\noindent with: $$\mathcal{M}_{2}:=\mathcal{L}_3-\mathcal{L}_2-\mathcal{R}_2+\mathcal{L}_{1}+C_2+\mathcal{R}_{3} \, .$$
and the construction collapses to the one we performed in Subsection \ref{sub3.1} for $S_{(3,2)} = R^{1,2,3}(3)$.
\label{remark1}
\end{rem}
	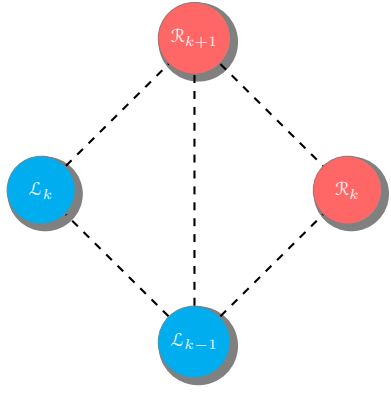
\begin{figure}[httb]
		\begin{center}
\begin{tikzpicture}[node distance=1.35cm and 1.35cm]
\node[state,style={minimum size=1cm, scale=0.9, circular drop shadow, fill=cyan,draw=gray,text=white}] (Lk)                             {\footnotesize{$\mathcal{L}_{k}$}};
\node[state,style={minimum size=1cm, scale=0.9, circular drop shadow, fill=red!60,draw=gray,text=white}] [above right=of Lk] (Rkp1)     {\footnotesize{$\mathcal{R}_{k+1}$}};
\node[state,style={minimum size=1cm,scale=0.9, circular drop shadow,fill=cyan,draw=gray,text=white}] [below right=of Lk] (Lk-1)      {\footnotesize{$\mathcal{L}_{k-1}$}};
\node[state,style={minimum size=1cm, scale=0.9, circular drop shadow, fill=red!60,draw=gray,text=white}] [below right=of Rkp1] (Rk) {\footnotesize{$\mathcal{R}_k$}};

\path [style={draw=black, dashed,thick}](Lk) edge node {} (Rkp1);
\path [style={draw=red!60, thick, dashed}](Rkp1) edge node {} (Rk);
\path [style={draw=black, dashed, thick}](Lk-1) edge node {} (Rk);
\path [style={draw=cyan, thick,  dashed}](Lk-1) edge node {} (Lk);
\path [style={draw=black,dashed, thick}](Lk-1) edge node {} (Rkp1);
\label{substructure}
\end{tikzpicture}
\caption{the generic $k$-th substructure $S_{(n,k)}$. We have not reported  $\mathcal{L}_n=\mathcal{R}_1$, which is central element of all $n-2$ substructures, as well as the one index central element $C_k$. Remember also that the generator $\mathcal{M}_k$ is a linear combination of the others (by virtue of  \eqref{Mk}).}
\label{fig1}
\end{center}
\end{figure}

\noindent Just as an explicit example, let us consider the $n=4$ case, so that two substructures arise: $S_{(4,2)}$ and $S_{(4,3)}$. The former,  which emerges with the choice of the three disjoint subsets $K_1^{(2)}=\{1\}$, $K_2^{(2)}=\{2\}$, $K_3^{(2)}=\{3,4\}$, is generated by the elements:
\begin{align}
\{\theta(C_1),\theta(C_2), \theta(C_3), \theta(C_{12}), \theta(C_{13}), \theta(C_{23}), \theta(C_{123})\} &\equiv \{C_1, C_2, C_{34}, C_{12}, C_{134}, C_{234}, C_{1234} \} \nonumber \\
&=\{\mathcal{L}_1, C_2, \mathcal{R}_3, \mathcal{L}_2, \mathcal{M}_2, \mathcal{R}_2, \mathcal{L}_4\}  \, ,\label{k2g}
\end{align}

\noindent with $\mathcal{L}_4=\mathcal{R}_1$ and $\mathcal{M}_{2}=\mathcal{L}_4-\mathcal{L}_2-\mathcal{R}_2+\mathcal{L}_1+C_2+\mathcal{R}_{3}$ by virtue of \eqref{Mk}.

\noindent The latter, which emerges instead by considering the disjoint subsets $K_1^{(3)}=\{1,2\}$, $K_2^{(3)}=\{3\}$, $K_3^{(3)}=\{4\}$, is generated by the elements:

\begin{align}
\{\theta(C_1),\theta(C_2), \theta(C_3), \theta(C_{12}), \theta(C_{13}), \theta(C_{23}), \theta(C_{123})\} &\equiv \{C_{12}, C_{3}, C_{4}, C_{123}, C_{124}, C_{34}, C_{1234} \} \nonumber \\
&=\{\mathcal{L}_2, C_3, \mathcal{R}_4,\mathcal{L}_3, \mathcal{M}_3, \mathcal{R}_3, \mathcal{L}_4\} \, , \label{k3g}
\end{align}
\noindent where $\mathcal{M}_3=\mathcal{L}_4-\mathcal{L}_3-\mathcal{R}_3+\mathcal{L}_2+C_3+\mathcal{R}_4$. Observing the two sets, we can see that the element $\mathcal{R}_3$, which appears as a central element of $S_{(4,2)}$, appears at the same time as a generator of $S_{(4,3)}$. The same happens for $\mathcal{L}_{2}$, which is central of $S_{(4,3)}$ and a generator of $S_{(4,2)}$.
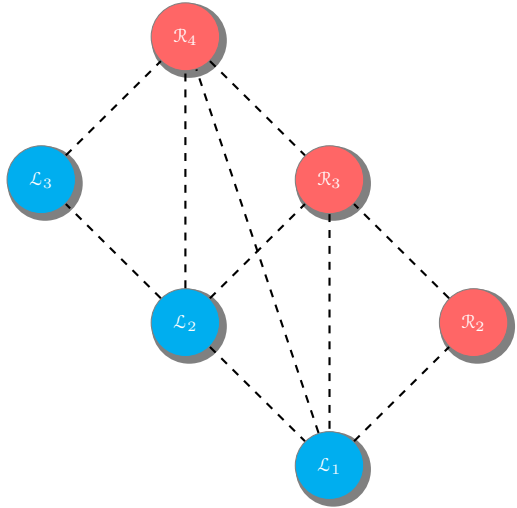
\begin{figure}[httb]
\begin{center}
\begin{tikzpicture}[node distance=1.25 cm and 1.25 cm]
\node[state, style={minimum size=1cm, scale=0.9, circular drop shadow, fill=cyan,draw=gray,text=white}] (L3)   {\footnotesize{$\mathcal{L}_{3}$}};
\node[state,style={minimum size=1cm, scale=0.9, circular drop shadow, fill=red!60,draw=gray,text=white}] [above right=of L3] (R4)     {\footnotesize{$\mathcal{R}_4$}};
\node[state,style={minimum size=1cm,scale=0.9, circular drop shadow,fill=cyan,draw=gray,text=white}] [below right=of L3] (L2)      {\footnotesize{$\mathcal{L}_{2}$}};
\node[state,style={minimum size=1cm, scale=0.9, circular drop shadow, fill=red!60,draw=gray,text=white}] [below right=of R4] (R3) {\footnotesize{$\mathcal{R}_3$}};
\node[state,style={minimum size=1cm, scale=0.9, circular drop shadow, fill=red!60,draw=gray,text=white}] [below right=of R3] (R2) {\footnotesize{$\mathcal{R}_2$}};
\node[state,style={minimum size=1cm, scale=0.9, circular drop shadow, fill=cyan,draw=gray,text=white}] [below right=of L2] (L1) {\footnotesize{$\mathcal{L}_1$}};
\node[state,style={minimum size=1cm, scale=0.9, fill=none,draw=none,text=white}] [below left=of L2] (V1) {};
\node[state,style={minimum size=1cm, scale=0.9, fill=none,draw=none,text=white}] [above right=of R3] (V2) {};
     	
     	\path [style={draw=black, dashed, thick}](L3) edge node {} (R4);
		\path [style={draw=red!60, thick, dashed}](R4) edge node {} (R3);
		\path [style={draw=black, dashed, thick}](L2) edge node {} (R3);
		\path [style={draw=cyan, thick, dashed}](L2) edge node {} (L3);
		\path [style={draw=black, dashed, thick}](L2) edge node {} (R4);
		\path [style={draw=cyan, thick, dashed}](L2) edge node {} (L1);
		\path [style={draw=black, dashed, thick}](L1) edge node {} (R2);
		\path [style={draw=red!60, thick, dashed}](R2) edge node {} (R3);
		\path [style={draw=black, dashed, thick}](R3) edge node {} (L1);
		\path [style={draw=black, dashed, thick}](L1) edge node {} (R4);

\end{tikzpicture}
	\caption{chain of quadratic algebras generated by the substructures $S_{(4,2)}$ and $S_{(4,3)}$ in dimension $n=4$. We have not reported  $\mathcal{L}_4=\mathcal{R}_1$, which is central element of both substructures, as well as the one index central elements $C_2, C_3$.}
	\label{fig2}
\end{center}
\end{figure}

\noindent Observing \figref{fig2} we can appreciate that at a fixed $n=4$, the element $\mathcal{R}_{i+1}$ (for a fixed $i=1, 2, 3$) Poisson commute with all the elements $\mathcal{L}_j$ with $j=1, \dots, i$. This is a general feature in the chain of these quadratic algebras and holds for any $n \geq 3$. The reason behind this is that the $(i+1)$-th right Casimir element has no common indices with the $i$ left Casimirs elements $\mathcal{L}_j$. This is a direct consequence of their definition on $\mathfrak{A} \otimes \mathfrak{A} \otimes \dots^{n)} \otimes \mathfrak{A}$.

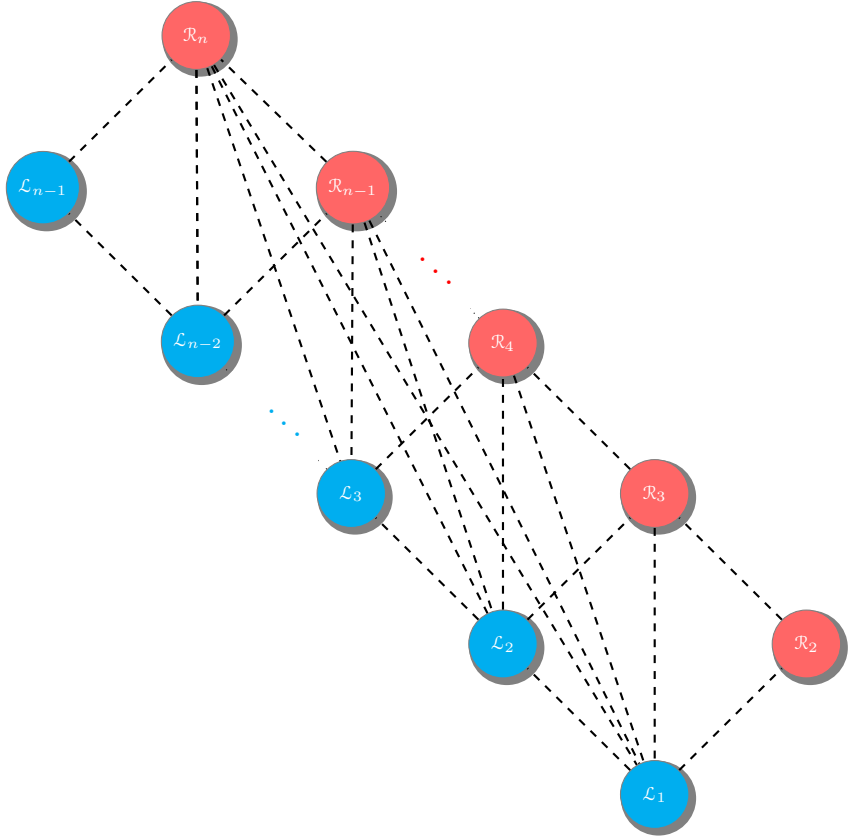
\begin{figure}[h!]
	\begin{center}
		\begin{tikzpicture}[node distance=1.35cm and 1.35 cm]
\node[state, style={minimum size=1cm, scale=0.9, circular drop shadow, fill=cyan,draw=gray,text=white}] (L1)  {\footnotesize{$\mathcal{L}_{1}$}};
\node[state,style={minimum size=1cm, scale=0.9, circular drop shadow, fill=red!60,draw=gray,text=white}] [above right=of	L1] (R2)     {\footnotesize{$\mathcal{R}_2$}};
\node[state,style={minimum size=1cm, scale=0.9, circular drop shadow, fill=cyan,draw=gray,text=white}] [above left=of	L1] (L2)     {\footnotesize{$\mathcal{L}_2$}};
\node[state,style={minimum size=1cm, scale=0.9, circular drop shadow, fill=red!60,draw=gray,text=white}] [above right=of	L2] (R3)     {\footnotesize{$\mathcal{R}_3$}};
\node[state,style={minimum size=1cm, scale=0.9, circular drop shadow, fill=cyan,draw=gray,text=white}] [above left=of	L2] (L3)     {\footnotesize{$\mathcal{L}_3$}};
\node[state,style={minimum size=1cm,scale=0.9, circular drop shadow,fill=red!60,draw=gray,text=white}] [above right=of	L3] (R4)      {\footnotesize{$\mathcal{R}_4$}};
\node[state,style={minimum size=1cm,scale=0.9, circular drop shadow,fill=cyan,draw=gray,text=white}] [above left=of L3] (Lnm2)      {\footnotesize{$\mathcal{L}_{n-2}$}};
\node[state,style={minimum size=1cm,scale=0.9, circular drop shadow,fill=red!60,draw=gray,text=white}] [above right=of Lnm2] (Rnm1)      {\footnotesize{$\mathcal{R}_{n-1}$}};
\node[state,style={minimum size=1cm,scale=0.9, circular drop shadow,fill=cyan,draw=gray,text=white}] [above left=of Lnm2] (Lnm1)      {\footnotesize{$\mathcal{L}_{n-1}$}};
\node[state,style={minimum size=1cm,scale=0.9, circular drop shadow,fill=red!60,draw=gray,text=white}] [above right=of Lnm1] (Rn)      {\footnotesize{$\mathcal{R}_{n}$}};		
\path [style={draw=black, dashed, thick}](L1) edge node {} (R2);
\path [style={draw=cyan,  thick, dashed}](L1) edge node {} (L2);
\path [style={draw=orange, dashed, thick}](L2) edge node {} (R3);
\path [style={draw=red!60,  thick, dashed}](R2) edge node {} (R3);
\path [style={draw=orange, dashed, thick}](R3) edge node {} (L1);
\path [style={draw=red!60,  thick, dashed}](R3) edge node {} (R4);
\path [style={draw=cyan,  thick, dashed}](L2) edge node {} (L3);
\path [style={draw=black, dashed, thick}](L3) edge node {} (R4);
\path [style={draw=black, dashed, thick}](L2) edge node {} (R4);
\path [style={draw=white, dotted }] (R4) edge node[anchor=center,fill=white, scale = 1.4] {$\small{\textcolor{white}{b}^{\color{red} {\ddots}}}$} (Rnm1);
\path [style={draw=white, loosely dotted}] (L3) edge node[anchor=center,fill=white, scale = 1.4] {$ \small{\textcolor{white}{a}^{\color{cyan} {\ddots}}}$} (Lnm2);
\path [style={draw=orange, dashed, thick}] (Rnm1) edge node {} (L3);
\path [style={draw=orange, dashed, thick}] (Rnm1) edge node {} (L2);
\path [style={draw=orange, dashed, thick}] (Rnm1) edge node {} (L1);
\path [style={draw=orange, dashed, thick}] (Rn) edge node {} (Lnm2);
\path [style={draw=black, dashed, thick}] (Rn) edge node {} (L3);
\path [style={draw=black, dashed, thick}] (Rn) edge node {} (L2);
\path [style={draw=black, dashed,thick }] (Rn) edge node {} (L1);
\path [style={draw=orange, dashed,thick}] (Lnm2) edge node {} (Rnm1);
\path [style={draw=cyan, thick, dashed}] (Lnm2) edge node {} (Lnm1);
\path [style={draw=black, dashed, thick}] (Lnm1) edge node {} (Rn);
\path [style={draw=red!60,thick, dashed}] (Rnm1) edge node {} (Rn);
\path [style={draw=black, thick, dashed}] (Rn) edge node {} (Lnm2);
\path [style={draw=black,dashed, thick}] (R4) edge node {} (L1);
 \label{substructures}
\end{tikzpicture}
\caption{chain of quadratic algebras composed by the $n-2$ substructures $S_{(n, k)}$. Again, the elements $\mathcal{L}_n=\mathcal{R}_1$ and $C_k$, with $k=2, \dots, n-1$, have been omitted.}
\label{fig3}
\end{center}
\end{figure}

\newpage \noindent If we focus on the generic $k$-th substructure in the chain, we can conclude that two elements, say $\{\mathcal{R}_k$, $\mathcal{R}_{k+1}\}$,  Poisson commute each other because belong to the set of right Casimirs (indicated in red), two other elements, $\{\mathcal{L}_k, \mathcal{L}_{k-1}\}$, Poisson commute each other because belong to the set of left Casimirs (indicated in blue) and the right element $R_{k+1}$ Poisson commute with the left ones in the $k$-th substructure because they have no common indices (this extends also to the left elements $\{\mathcal{L}_1, \dots, \mathcal{L}_{k-2}\}$ outside the $k$-th substructure).

To summarise, this is the chain of quadratic algebras generated by the left and right Casimir invariants. The main point of this result concerns the fact that behind the above chain structure, which we recall has been crucial in the quantum case to construct the spectrum of $n$D quantum superintegrable models algebraically \cite{10.1088/1751-8121/aac111}, there is in fact a chain of (Poisson) Racah algebras in disguise. They are generated by the elements in the set \eqref{eq:gene}. We conclude the discussion about the classical case observing that each of the $n-2$ substructures is endowed with a Casimir function. Specifically, at fixed $n \geq 3$, for each $k=2, \dots, n-1$ we can define the $k$-th Casimir associated to the $k$-th embedded Poisson algebra $R^{\{1, \dots, k-1\},\{k\},\{k+1, \dots, n\}}(3)$, which turns out to be:
\begin{equation}
\mathcal{K}_k = \mathcal{F}^2_{k}+\mathcal{L}_k \mathcal{M}_k \mathcal{R}_k+ (\mathcal{R}_{k+1}-C_k)(\mathcal{L}_{k-1}-\mathcal{L}_n)\mathcal{L}_k-(C_k-\mathcal{L}_{k-1})(\mathcal{R}_{k+1}-\mathcal{L}_n)\mathcal{R}_k \, .
\label{eq:casemb}
\end{equation}
By direct computation, at a fixed $n$ and $k$, it is verified that each Casimir Poisson commute with the corresponding elements in the set \eqref{eq:gene}.
\section{Quadratic substructures generated by the left and right Casimirs as rank one Racah subalgebras: The quantum case}
\label{sec4}
\noindent Our starting point for the quantum case is the Lie coalgebra $\mathfrak{A}$ in the basis $(J_\pm, J_3)$ with commutation rules:
\begin{equation}
	[J_-,J_+]= 2 \imath \hbar J_3 \quad [J_3, J_{\pm}]=\pm \imath \hbar  J_\pm \, ,
	\label{eq:Lie}
\end{equation}
endowed with the usual primitive coproduct $\Delta: \mathfrak{A} \to  \mathfrak{A} \otimes  \mathfrak{A} $ and non-linear Casimir:
\begin{equation}
	C (J_+,J_-,J_3)= J_3^2-\frac{1}{2}(J_+J_-+J_-J_+) \, .
	\label{nonlin}
\end{equation}

\noindent A one-dimensional differential realisation $\bar{D}$ for this Lie coalgebra is:
\begin{equation}
	\hat{J}_+^{[1]} := \bar{D}(J_+)= \frac{1}{2}\bigl(-\hbar^2\partial_{x_1}^2+\frac{a_1}{x_1^2}  \bigl)\quad \hat{J}_-^{[1]}: =\bar{D}(J_-)= \frac{1}{2}x_1^2\quad \hat{J}_3^{[1]} :=\bar{D}(J_3)=  -\frac{\imath \hbar}{2} (x_1 \cdot \partial_{x_1}+1/2) \, .
	\label{1Dqreal}
\end{equation}
In the given realisation, the Casimir operator reads:
\begin{equation}
	\hat{C}^{[1]}:=\bar{D}(C(J_+,J_-,J_3))=\frac{1}{16}(3 \hbar^2-4 a_1) \, .
	\label{eq:cas1reaq}
\end{equation}
A $n$D realisation is then given by $n$ copies of it:
\begin{equation}
\hat{J}_+^{[n]} = \frac{1}{2}\biggl(-\hbar^2\Delta+\sum_{j=1}^n\frac{a_j}{x_j^2} \biggl)\quad \hat{J}_-^{[n]} = \frac{1}{2}\bx^2\quad \hat{J}_3^{[n]} = -\frac{\imath \hbar}{2} (\bx \cdot \nabla+n/2) \, ,
\label{eq:nD}
\end{equation}
and, in total analogy to the classical case, the left and right partial Casimirs read:
\begin{equation}
\hat{C}^{[m]}=\sum_{1 \leq i< j}^m \hat{C}_{ij}-(m-2)\sum_{i=1}^m \hat{C}_i \qquad \hat{C}_{[m]}=\sum_{n-m+1 \leq i< j}^n \hat{C}_{ij}-(m-2)\sum_{i=n-m+1}^n \hat{C}_i \, ,
\label{lrcasq}
\end{equation}
\noindent where we introduced the quantities (the quantum analog of \eqref{eq:racahgen}):
\begin{equation}
\hat{C}_{ij}:=-\frac{1}{4}\biggl( \hat{L}_{ij}^2+a_i \frac{\hat{x}_j^2}{\hat{x}_i^2}+a_j \frac{\hat{x}_i^2}{\hat{x}_j^2} +a_i+a_j-\hbar^2\biggl) \qquad \hat{C}_{i}:= \frac{1}{16}(3 \hbar^2-4 a_i) \, ,
\label{Racahgen}
\end{equation}
with $\hat{L}_{ij} = \hat{x}_i \hat{p}_j -\hat{x}_j \hat{p}_i$ ($\hat{p}_i \equiv -\imath \hbar \partial_{x_i}$). As expected, we see that the two expressions \eqref{lrcasq} are formally equivalent to the ones in \eqref{lrRac}, what changes is just the different realisation used in the quantum case for the generators. Again, if we introduce the operators: $\hat{P}_{ij}:=\hat{C}_{ij}-\hat{C}_i-\hat{C}_j$ the Casimirs can be expressed as:
\begin{equation}
\hat{C}^{[m]}=\sum_{1 \leq i< j}^m \hat{P}_{ij}+\sum_{i=1}^m \hat{C}_i \, ,\qquad \hat{C}_{[m]}=\sum_{n-m+1 \leq i< j}^n \hat{P}_{ij}+\sum_{i=n-m+1}^n \hat{C}_i
\label{eq:qcas}
\end{equation}
and, with the new definitions $\hat{\mathcal{L}}_k := \hat{C}^{[k]}$ and $\hat{\mathcal{R}}_k := \hat{C}_{[n-k+1]}$, for $k=1, \dots, n$, we can construct the quantum analog of \eqref{lrpij}, which read:
\begin{equation}
\hat{\mathcal{L}}_k =\sum_{1 \leq i< j}^k \hat{P}_{ij}+\sum_{i=1}^k \hat{C}_i \, ,  \qquad  \hat{\mathcal{R}}_k := \sum_{k \leq i< j}^n \hat{P}_{ij}+\sum_{i=k}^n \hat{C}_i  \, .
\label{qanaloglr}
\end{equation}
Clearly, because of its underlying algebraic nature, the construction is very much the same as the one we performed for the classical case.  We need to take into account, besides the different realisation, the noncommutativity of the product of operators. In what follow, we limit ourselves to review the construction for the $n=3$ case and to rephrase Proposition \ref{prop1} in the quantum framework.
\subsection{The $\boldsymbol{n=3}$ case: the rank one Racah algebra $\boldsymbol{R(3)}$}
\label{sub4.1}
\noindent The generators of the  rank one Racah algebra are: $\{\hat{C}_1, \hat{C}_2, \hat{C}_3, \hat{C}_{12}, \hat{C}_{13}, \hat{C}_{23}, \hat{C}_{123} \}$ with:
\begin{equation}
\hat{C}_{123}=\hat{C}_{12}+\hat{C}_{13}+\hat{C}_{23}-\hat{C}_1-\hat{C}_2-\hat{C}_3  \, ,
\label{lineareqq}
\end{equation}
where, as asual, $\hat{C}_1$, $\hat{C}_2$, $\hat{C}_3$ and $\hat{C}_{123}$ are central elements.
\noindent Now, by defining:
\begin{equation}
\hat{\mathcal{F}}:=\frac{1}{2 \imath \hbar}[\hat{C}_{12},\hat{C}_{23}]=\frac{1}{2 \imath \hbar}[\hat{C}_{23},\hat{C}_{13}]=\frac{1}{2 \imath \hbar}[\hat{C}_{13},\hat{C}_{12}]
\label{funFq}
\end{equation}
\noindent we get the following defining commutation relations:
\begin{align}
&[\hat{C}_{12},\hat{\mathcal{F}}]/\imath \hbar=\hat{C}_{23}\hat{C}_{12}-\hat{C}_{12}\hat{C}_{13}+(\hat{C}_2-\hat{C}_1)(\hat{C}_3-\hat{C}_{123}) \label{1racq}\\
&[\hat{C}_{23},\hat{\mathcal{F}}]/\imath \hbar=\hat{C}_{13}\hat{C}_{23}-\hat{C}_{23}\hat{C}_{12}+(\hat{C}_3-\hat{C}_2)(\hat{C}_1-\hat{C}_{123})  \label{2racq}\\
&[\hat{C}_{13},\hat{\mathcal{F}}]/\imath \hbar=\hat{C}_{12}\hat{C}_{13}-\hat{C}_{13}\hat{C}_{23}+(\hat{C}_1-\hat{C}_3)(\hat{C}_2-\hat{C}_{123}) \, .  \label{3racq}
\end{align}
From \eqref{lineareqq}, using the properties of the commutators, we obtain the relations:
\begin{equation}
[\hat{C}_{23}, \hat{C}_{12}+\hat{C}_{13}]=0 \qquad [\hat{C}_{12}, \hat{C}_{13}+\hat{C}_{23}]=0 \qquad [\hat{C}_{13}, \hat{C}_{12}+\hat{C}_{23}]=0 \, .
\label{eq:DKq}
\end{equation}
Also, the following sum is easily checked to be zero:
\begin{equation}
[\hat{C}_{12},\hat{\mathcal{F}}]+[\hat{C}_{23},\hat{\mathcal{F}}]+[\hat{C}_{13},\hat{\mathcal{F}}]=0 \, .
\label{sumq}
\end{equation}
\noindent Again, if we use \eqref{lineareqq} to explicitate the $\hat{C}_{13}$ generator, i.e.:
\begin{equation}
\hat{C}_{13}=\hat{C}_{123}-\hat{C}_{12}-\hat{C}_{23}+\hat{C}_1+\hat{C}_2+\hat{C}_3 \,
\label{c13q}
\end{equation}
and we substitute this expressions into  \eqref{1racq}-\eqref{2racq}, taking into account \eqref{sumq}, we get:
\begin{align}
&[\hat{C}_{12},\hat{\mathcal{F}}]/\imath \hbar= +\hat{C}_{12}^2+\{\hat{C}_{23},\hat{C}_{12}\}-(\hat{C}_{123}+\hat{C}_1+\hat{C}_2+\hat{C}_3)\hat{C}_{12}+(\hat{C}_2-\hat{C}_1)(\hat{C}_3-\hat{C}_{123})\\
&[\hat{C}_{23},\hat{\mathcal{F}}]/\imath \hbar= -\hat{C}_{23}^2-\{\hat{C}_{12},\hat{C}_{23}\}+(\hat{C}_{123}+\hat{C}_1+\hat{C}_2+\hat{C}_3)\hat{C}_{23}+(\hat{C}_3-\hat{C}_2)(\hat{C}_1-\hat{C}_{123}) \, ,
\end{align}
which represents the quantum analog of  \eqref{eq:sub321}-\eqref{eq:sub322}. Here $\{\hat{A},\hat{B}\}:=\hat{A}\hat{B}+\hat{B}\hat{A}$ is the anti-commutator. In total analogy to the classical case, we can then rewrite the substructure $(n,k)=(3,2)$, which we indicate as $S^{(\hbar)}_{(3,2)}$, as:
\begin{align}
&[\hat{\mathcal{L}}_2,\hat{\mathcal{F}}]/\imath \hbar=+\hat{\mathcal{L}}_2^2+\{\hat{\mathcal{R}}_2, \mathcal{\hat{L}}_2\}-(\mathcal{\hat{L}}_3+\hat{\mathcal{L}}_{1}+\hat{C}_2+\hat{\mathcal{R}}_{3})\hat{\mathcal{L}}_2+(\hat{C}_2-\hat{\mathcal{L}}_{1})(\hat{\mathcal{R}}_{3}-\hat{\mathcal{L}}_3) \label{s32q1}\\
&[\hat{\mathcal{R}}_2,\hat{\mathcal{F}}]/\imath \hbar=-\hat{\mathcal{R}}_2^2-\{\hat{\mathcal{L}}_2, \hat{\mathcal{R}}_2\}+(\hat{\mathcal{L}}_3+\hat{\mathcal{L}}_{1}+\hat{C}_2+\hat{\mathcal{R}}_{3})\hat{\mathcal{R}}_2+(\hat{\mathcal{R}}_{3}-\hat{C}_{2})(\hat{\mathcal{L}}_{1}-\hat{\mathcal{L}}_3) \label{s32q2} \, ,
\end{align}
with:
\begin{equation}
\hat{\mathcal{F}}:=\frac{1}{2 \imath \hbar}[\hat{\mathcal{L}}_2,\hat{\mathcal{R}}_2]=\frac{1}{2 \imath \hbar}[\hat{\mathcal{R}}_2,\hat{\mathcal{M}}_2]=\frac{1}{2 \imath \hbar}[\hat{\mathcal{M}}_2,\hat{\mathcal{L}}_2]\, ,
\label{FLRijkq}
\end{equation}
and the usual definition $\hat{\mathcal{M}}_2:=\hat{\mathcal{L}}_3-\hat{\mathcal{L}}_2-\hat{\mathcal{R}}_2+\hat{\mathcal{L}}_1+\hat{C}_2+\hat{\mathcal{R}}_3$ for the generator $\hat{C}_{13}$. The relations \eqref{s32q1}-\eqref{s32q2}-\eqref{FLRijkq} represent the quantum analog of \eqref{s32c1}-\eqref{s32c2}-\eqref{FLRijk}.

\subsection{Higher dimensional quantum 
substructures, injective morphisms and rank one Racah algebras $\boldsymbol{R^{K_1,K_2,K_3}(3)}$}
\label{sub4.2}
\noindent Once the link with the Racah algebra $R(3)$ has been found, we can proceed with an analogous construction for the quantum substructures. In particular, we can formulate the following:
\begin{prop}
\label{prop2}
Let $[n]:=\{1, \dots, n\}$ be a set with $n>3$ fixed. For each $k=2, \dots, n-1$ define $K_1^{(k)}:=\{1,\dots, k-1\}$, $K_2^{(k):}=\{k\}$ and $K_3^{(k)}:=\{k+1, \dots, n\}$ to be three disjoint subsets of the set $[n]$. Then, the images of the $n-2$ injective morphisms $\theta_k:$  $R(3) \to R(n)$ $(k=2, \dots, n-1)$:
\begin{align*}
&\theta_k(C_1)=C_{K_1^{(k)}} \equiv \hat{C}_{1 \dots k-1}=\hat{\mathcal{L}}_{k-1}\\
&\theta_k(C_2)=C_{K_2^{(k)}}\equiv\hat{C}_k\ \\
&\theta_k(C_3)=C_{K_3^{(k)}}\equiv\hat{C}_{k+1 \dots n}=\hat{\mathcal{R}}_{k+1}\\
&\theta_k(C_{12})=C_{K_1^{(k)} K_2^{(k)}}\equiv\hat{C}_{1 \dots k}=\hat{\mathcal{L}}_k\\
&\theta_k(C_{23})=C_{K_2^{(k)} K_3^{(k)}}\equiv\hat{C}_{k \dots n}=\hat{\mathcal{R}}_k\\
&\theta_k(C_{13})=C_{K_1^{(k)} K_3^{(k)}}\equiv\hat{C}_{1 \dots k-1 k+1 \dots n}=\hat{\mathcal{M}}_k\\
&\theta_k(C_{123})=C_{K_1^{(k)} K_2^{(k)} K_3^{(k)}}\equiv\hat{C}_{1\dots n}=\hat{\mathcal{L}}_n=\hat{\mathcal{R}}_1
\end{align*}
\noindent with: $$\hat{\mathcal{M}}_{k}:=\hat{\mathcal{L}}_n-\hat{\mathcal{L}}_k-\hat{\mathcal{R}}_k+\hat{\mathcal{L}}_{k-1}+\hat{C}_k+\hat{\mathcal{R}}_{k+1} \, ,$$
result in the universal quantum substructures $S^{(\hbar)}_{(n,k)}$. Moreover, as a direct consequence of Lemma \ref{Lemma} they are isomorphic to the rank one Racah algebra $R(3)$.
\end{prop}
\begin{proof}
For each $k=2, \dots, n-1$, let $K_1^{(k)}:=\{1,\dots, k-1\}$, $K_2^{(k):}=\{k\}$ and $K_3^{(k)}:=\{k+1, \dots, n\}$ be three specific disjoint subsets of $[n]$. With this choice, the quantum generators are:
	\begin{equation}
	\{\hat{\mathcal{L}}_{k-1}, \hat{C}_k, \hat{\mathcal{R}}_{k+1},  \hat{\mathcal{L}}_k,  \hat{\mathcal{R}}_k,  \hat{\mathcal{M}}_k,  \hat{\mathcal{L}}_n= \hat{\mathcal{R}}_1\},
	\label{eq:geneq}
	\end{equation}
$\hat{\mathcal{L}}_{k-1}$, $\hat{C}_k$, $\hat{\mathcal{R}}_{k+1}$ and $\hat{\mathcal{L}}_n=\hat{\mathcal{R}}_1$ being central elements associated to the $k$-th substructure.	The linear relation \eqref{c13q} is lifted to the $n-2$ linear relations:
	\begin{equation}
	\hat{\mathcal{L}}_n = \hat{\mathcal{L}}_{k}+\hat{\mathcal{R}}_{k}+\hat{\mathcal{M}}_{k}-\hat{\mathcal{L}}_{k-1}-\hat{C}_k-\hat{\mathcal{R}}_{k+1}=\hat{\mathcal{R}}_1 \, ,
	\label{lineqq}
	\end{equation}
and the rank one Racah subalgebras turn out to be:
\begin{align}
&[\hat{\mathcal{L}}_k,\hat{\mathcal{F}}_k]/\imath \hbar\,=\hat{\mathcal{R}}_k\hat{\mathcal{L}}_k-\hat{\mathcal{L}}_k \hat{\mathcal{M}}_k\,+(\hat{C}_k-\hat{\mathcal{L}}_{k-1})(\hat{\mathcal{R}}_{k+1}-\hat{\mathcal{L}}_n) \label{racahq1}\\
&[\hat{\mathcal{R}}_k,\hat{\mathcal{F}}_k]/\imath \hbar\,=\hat{\mathcal{M}}_k\hat{\mathcal{R}}_k-\hat{\mathcal{R}}_k\hat{\mathcal{L}}_k\,+(\hat{\mathcal{R}}_{k+1}-\hat{C}_k)(\hat{\mathcal{L}}_{k-1}-\hat{\mathcal{L}}_n) \label{racahq2}\\
&[\hat{\mathcal{M}}_k,\hat{\mathcal{F}}_k]/\imath \hbar=\hat{\mathcal{L}}_k\hat{\mathcal{M}}_k-\hat{\mathcal{M}}_k\hat{\mathcal{R}}_k+(\hat{\mathcal{L}}_{k-1}-\hat{\mathcal{R}}_{k+1})(\hat{C}_k-\hat{\mathcal{L}}_{n}) \label{racahq3} \, ,
\end{align}
	where:
	\begin{equation}
	\hat{\mathcal{F}}_k:=\frac{1}{2 \imath \hbar}[\hat{\mathcal{L}}_k,\hat{\mathcal{R}}_k]=\frac{1}{2 \imath \hbar}[\hat{\mathcal{R}}_k,\hat{\mathcal{M}}_{k}]=\frac{1}{2 \imath \hbar}[\hat{\mathcal{M}}_{k},\hat{\mathcal{L}}_k] \, .
	\label{funq}
	\end{equation}
	The relations given in \eqref{eq:DKq}, are lifted to the following ones:
	\begin{equation}
	[\hat{\mathcal{R}}_k, \hat{\mathcal{L}}_{k}+\hat{\mathcal{M}}_{k}]=0 \, , \qquad [\hat{\mathcal{L}}_k, \hat{\mathcal{M}}_{k}+\hat{\mathcal{R}}_{k}]=0 \, , \qquad [\hat{\mathcal{M}}_{k}, \hat{\mathcal{L}}_{k}+\hat{\mathcal{R}}_{k}]=0 \, .
	\label{DKq}
	\end{equation}
	Moreover, relation \eqref{sumq}, is now rephrased as:
	
	\begin{equation}
	[\hat{\mathcal{L}}_k,\hat{\mathcal{F}}_k]+[\hat{\mathcal{R}}_k, \hat{\mathcal{F}}_k]+[\hat{\mathcal{M}}_k, \hat{\mathcal{F}}_k]=0 \,.
	\label{sumqq}
	\end{equation}
	In total analogy to the three dimensional case, we use \eqref{lineqq} to explicitate the $\hat{\mathcal{M}}_k$ generators, i.e.:
	\begin{equation}
	\hat{\mathcal{M}}_{k}=	\hat{\mathcal{L}}_n-	\hat{\mathcal{L}}_k-	\hat{\mathcal{R}}_k+	\hat{\mathcal{L}}_{k-1}+	\hat{C}_k+	\hat{\mathcal{R}}_{k+1} \,
	\label{Mkq}
	\end{equation}
	and we substitute these expressions into  \eqref{racahq1}-\eqref{racahq2}, taking into account \eqref{sumqq}. In this way we get:
	\begin{align}
	&[\hat{\mathcal{L}}_k,\hat{\mathcal{F}}_k]/\imath \hbar=+\hat{\mathcal{L}}_k^2+\{\hat{\mathcal{R}}_k, \hat{\mathcal{L}}_k\}-(\hat{\mathcal{L}}_n+\hat{\mathcal{L}}_{k-1}+\hat{C}_k+\hat{\mathcal{R}}_{k+1})\hat{\mathcal{L}}_k+(\hat{C}_k-\hat{\mathcal{L}}_{k-1})(\hat{\mathcal{R}}_{k+1}-\hat{\mathcal{L}}_n) \label{qsubstructure1}\\
	&[\mathcal{R}_k,\hat{\mathcal{F}}_k]/\imath \hbar=-\hat{\mathcal{R}}_k^2-\{\hat{\mathcal{L}}_k, \hat{\mathcal{R}}_k\}+(\hat{\mathcal{L}}_n+\hat{\mathcal{L}}_{k-1}+\hat{C}_k+\hat{\mathcal{R}}_{k+1})\hat{\mathcal{R}}_k+(\hat{\mathcal{R}}_{k+1}-\hat{C}_{k})(\hat{\mathcal{L}}_{k-1}-\hat{\mathcal{L}}_n) \label{qsubstructure2} \, ,
	\end{align}
	\noindent where $\hat{\mathcal{L}}_n=\hat{C}^{[n]}=\hat{C}_{[n]}=\hat{\mathcal{R}}_1$. At a fixed $n$, for each $k=2, \dots, n-1$, these commutation relations define the universal quadratic substructures $S^{(\hbar)}_{(n,k)}$ generated by the quantum left and right Casimir invariants.
\end{proof}
\begin{rem}
	The result obtained in the classical case is recovered by performing the appropriate $\hbar \to 0$ limit. In this limit, in fact,  \eqref{csubstructure1}-\eqref{csubstructure2} arise from \eqref{qsubstructure1}-\eqref{qsubstructure2} by replacing the commutator $[\hat{A}, \hat{B}]/\imath \hbar$ with the Poisson bracket $\{A, B\}$ and by taking into account that $\{\hat{A},\hat{B}\} \to 2 AB$. In short, $\lim_{\hbar \to 0} S_{(n,k)}^{(\hbar)} = S_{(n,k)}$. This holds true also for all the other commutation relations appearing in the construction.
\end{rem}
\noindent Also in the quantum setting it is possible to construct the $k$-th Casimir associated to the $k$-th substructure. It can be cast in the form:
\begin{align}
\hat{\mathcal{K}}_k = \hat{\mathcal{F}}^2_{k}&+\frac{1}{6}\{\hat{\mathcal{L}}_k, \hat{\mathcal{M}}_k, \hat{\mathcal{R}}_k\}+ (\hat{\mathcal{R}}_{k+1}-\hat{C}_k)(\hat{\mathcal{L}}_{k-1}-\hat{\mathcal{L}}_n)\hat{\mathcal{L}}_k-(\hat{C}_k-\hat{\mathcal{L}}_{k-1})(\hat{\mathcal{R}}_{k+1}-\hat{\mathcal{L}}_n)\hat{\mathcal{R}}_k \nonumber\\
&+\frac{\hbar^2}{3}(\{\hat{\mathcal{L}}_k,\hat{\mathcal{M}}_k\}+\{\hat{\mathcal{L}}_k,\hat{\mathcal{R}}_k\}+\{\hat{\mathcal{M}}_k,\hat{\mathcal{R}}_k\}+(\hat{C}_k-\hat{\mathcal{L}}_{k-1})(\hat{\mathcal{R}}_{k+1}-\hat{\mathcal{L}}_n)-(\hat{\mathcal{R}}_{k+1}-\hat{C}_k)(\hat{\mathcal{L}}_{k-1}-\hat{\mathcal{L}}_n)) \, ,
\label{eq:casq}
\end{align}
where we introduced the symmetrizer of three operators $\{\hat{A},\hat{B},\hat{C}\} := \hat{A}\hat{B}\hat{C}+\hat{A}\hat{C}\hat{B}+\hat{B}\hat{A}\hat{C}+\hat{B}\hat{C}\hat{A}+\hat{C}\hat{A}\hat{B}+\hat{C}\hat{B}\hat{A}$. By direct computation, at a fixed $n$ and $k$, it is verified that each of the above $n-2$ Casimir invariants commutes with the corresponding elements in the set \eqref{eq:geneq}.

\section{Concluding Remarks}
\label{sec5}
\noindent In this paper we have reexamined the Racah algebra $R(n)$, a rank-($n-2)$ quadratic algebra, from the perspective of the subalgebra structures which are useful in the algebraic derivation of spectrum for $n$-dimensional superintegrable systems \cite{10.1088/1751-8121/aac111, Latini_2019, latini2020embedding}. We have determined the commutation relations of multi-indexed generators of the Racah algebra and related these generators of substructures to the left and right Casimirs which arise from coalgebra symmetry. In particular, we have showed that the $n-2$ classical/quantum quadratic substructures generated by the left and right classical/quantum integrals, coming from the application of the left and right $m$-th coproduct maps to the “seed”Casimir of the coalgebra, can be understood as a classical/quantum realisation of $n - 2$ rank one Racah algebras. We explicitly shown the role played by the family of injective morphisms in the construction of the universal substructures. As a byproduct, we have obtained a new presentation for each of the $n - 2$ quadratic algebras generated by the left and right integrals.

Our results shed new light on the idea of substructures which was originally introduced by Daskaloyannis in \cite{2011SIGMA...7..054T} for a $3$-dimensional system and then applied to models on $n$-dimensional Euclidean spaces as well as on spaces of constant curvatures \cite{Latini_2019, 10.1088/1751-8121/aac111}. The advantage of this approach for calculating spectrum is that it relies only on certain sets of constraints. The present work allows us to understand more generally how to construct substructures and algebraically derive the spectra for other superintegrable systems in dimension greater than 2, which is a difficult problem as our knowledge to the representation theory of higher rank quadratic algebras is still rather limited.

These ideas would also be applicable with certain modifications to models obtained from  block separation of variables \cite{Chen_2019}. Such models possess quadratic algebra structures with structure constants depending on additional central elements given by Casimirs of some higher rank Lie algebras.

\section*{Acknowledgement}
IM was supported by Australian Research Council Future Fellowship FT180100099. YZZ was supported by Australian Research Council Discovery Project DP190101529 and National Natural Science Foundation of China (Grant No. 11775177).

\addcontentsline{toc}{chapter}{Bibliography}
\bibliographystyle{utphys}
\bibliography{bibliography}

\providecommand{\href}[2]{#2}\begingroup\raggedright\begin{thebibliography}{10}

\bibitem{MillerPostWinternitz2013R}
W.~J. Miller, S.~Post, and P.~Winternitz, ``{C}lassical and {Q}uantum
  {S}uperintegrability with {A}pplications.''
  \href{http://dx.doi.org/10.1088/1751-8113/46/42/423001}{{\em J. Phys. A:
  Math. Theor.} {\bfseries 46} no.~42, (2013) 423001}.

\bibitem{2014RCD....19..415B}
I.~A. {Bizyaev}, A.~V. {Borisov}, and I.~S. {Mamaev}, ``{S}uperintegrable
  {G}eneralizations of the {K}epler and {H}ook problems.''
  \href{http://dx.doi.org/10.1134/S1560354714030095}{{\em {R}egular and
  {C}haotic {D}ynamics} {\bfseries 19} (2014) 415--434}.

\bibitem{Hoque_2015}
M.~F. Hoque, I.~Marquette, and Y.-Z. Zhang, ``Quadratic algebra structure and
  spectrum of a new superintegrable system in ${N}$-dimension.''
  \href{http://dx.doi.org/10.1088/1751-8113/48/18/185201}{{\em Journal of
  Physics A: Mathematical and Theoretical} {\bfseries 48} no.~18, (2015)
  185201}.

\bibitem{1751-8121-49-12-125201}
M.~F. Hoque, I.~Marquette, and Y.-Z. Zhang, ``Recurrence approach and higher
  rank cubic algebras for the ${N}$-dimensional superintegrable systems.''
  \href{http://dx.doi.org/10.1088/1751-8113/49/12/125201}{{\em Journal of
  Physics A: Mathematical and Theoretical} {\bfseries 49} no.~12, (2016)
  125201}.

\bibitem{Hoque_2015_}
M.~F. Hoque, I.~Marquette, and Y.-Z. Zhang, ``A new family of ${N}$-dimensional
  superintegrable double singular oscillators and quadratic algebra {Q}(3)
  $\oplus$ so($n$) $\oplus$ so(${N}$-$n$).''
  \href{http://dx.doi.org/10.1088/1751-8113/48/44/445207}{{\em Journal of
  Physics A: Mathematical and Theoretical} {\bfseries 48} no.~44, (2015)
  445207}.

\bibitem{Hoque2018}
M.~F. Hoque, I.~Marquette, and Y.-Z. Zhang, ``On superintegrable monopole
  systems.'' \href{http://dx.doi.org/10.1088/1742-6596/965/1/012018}{{\em
  Journal of Physics: Conference Series} {\bfseries 965} (Feb, 2018) 012018}.

\bibitem{Chen_2019}
Z.~Chen, I.~Marquette, and Y.-Z. Zhang, ``Superintegrable systems from block
  separation of variables and unified derivation of their quadratic algebras.''
  \href{http://dx.doi.org/10.1016/j.aop.2019.167970}{{\em Annals of Physics}
  {\bfseries 411} (Dec, 2019) 167970}.

\bibitem{Chen_2019_}
Z.~Chen, I.~Marquette, and Y.-Z. Zhang, ``{E}xtended
  {L}aplace–{R}unge–{L}entz vectors, new family of superintegrable systems
  and quadratic algebras.''
  \href{http://dx.doi.org/10.1016/j.aop.2019.01.009}{{\em Annals of Physics}
  {\bfseries 402} (Mar, 2019) 78–90}.

\bibitem{Kalnins05_}
E.~G. Kalnins, J.~M. Kress, and W.~Miller, ``{S}econd-order superintegrable
  systems in conformally flat spaces. {I}. {T}wo-dimensional classical
  structure theory.'' \href{http://dx.doi.org/10.1063/1.1897183}{{\em Journal
  of Mathematical Physics} {\bfseries 46} no.~5, (2005) 053509}.

\bibitem{Kalnins05_2}
E.~G. Kalnins, J.~M. Kress, and W.~Miller, ``{S}econd order superintegrable
  systems in conformally flat spaces. {II}. {T}he classical two-dimensional
  stäckel transform.'' \href{http://dx.doi.org/10.1063/1.1894985}{{\em Journal
  of Mathematical Physics} {\bfseries 46} no.~5, (2005) 053510}.

\bibitem{Kalnins05_3}
E.~G. Kalnins, J.~M. Kress, and W.~Miller, ``{S}econd order superintegrable
  systems in conformally flat spaces. {III}. {T}hree-dimensional classical
  structure theory.'' \href{http://dx.doi.org/10.1063/1.2037567}{{\em Journal
  of Mathematical Physics} {\bfseries 46} no.~10, (2005) 103507}.

\bibitem{Kalnins06_4}
E.~G. Kalnins, J.~M. Kress, and W.~Miller, ``{S}econd order superintegrable
  systems in conformally flat spaces. {I}{V}. {T}he classical 3{D} {S}täckel
  transform and 3{D} classification theory.''
  \href{http://dx.doi.org/10.1063/1.2191789}{{\em Journal of Mathematical
  Physics} {\bfseries 47} no.~4, (2006) 043514}.

\bibitem{Kalnins06_5}
E.~G. Kalnins, J.~M. Kress, and W.~Miller, ``{S}econd-order superintegrable
  systems in conformally flat spaces. {V}. {T}wo- and three-dimensional quantum
  systems.'' \href{http://dx.doi.org/10.1063/1.2337849}{{\em Journal of
  Mathematical Physics} {\bfseries 47} no.~9, (2006) 093501}.

\bibitem{Kalnins_2013}
E.~G. Kalnins, ``{C}ontractions of 2{D} 2nd {O}rder {Q}uantum {S}uperintegrable
  {S}ystems and the {A}skey {S}cheme for {H}ypergeometric {O}rthogonal
  {P}olynomials.'' \href{http://dx.doi.org/10.3842/sigma.2013.057}{{\em
  Symmetry, Integrability and Geometry: Methods and Applications} (2013) }.

\bibitem{De_Bie_2017}
H.~D. Bie, V.~X. Genest, W.~van~de Vijver, and L.~Vinet, ``{A} higher rank
  {R}acah algebra and the $\mathbb{Z}_2^n$ {L}aplace-{D}unkl operator.''
  \href{http://dx.doi.org/10.1088/1751-8121/aa9756}{{\em Journal of Physics A:
  Mathematical and Theoretical} {\bfseries 51} no.~2, (Dec, 2017) 025203}.

\bibitem{Iliev_2017}
P.~Iliev, ``The generic quantum superintegrable system on the sphere and
  {R}acah operators.'' \href{http://dx.doi.org/10.1007/s11005-017-0978-3}{{\em
  Letters in Mathematical Physics} {\bfseries 107} no.~11, (2017) 2029–2045}.

\bibitem{bie2020racah}
H.~{D}e {B}ie, P.~{I}liev, W.~{V}an~de {V}ijver, and L.~{V}inet, ``The {R}acah
  algebra: {A}n overview and recent results.''
  \href{http://arxiv.org/abs/2001.11195}{{\ttfamily arXiv:2001.11195
  [math.RT]}}.

\bibitem{Gaboriaud_2019}
J.~Gaboriaud, L.~Vinet, S.~Vinet, and A.~Zhedanov, ``The generalized {R}acah
  algebra as a commutant.''
  \href{http://dx.doi.org/10.1088/1742-6596/1194/1/012034}{{\em Journal of
  Physics: Conference Series} {\bfseries 1194} (2019) 012034}.

\bibitem{Crampe2020}
N.~Cramp{\'e}, W.~van~de Vijver, and L.~Vinet, ``{R}acah {P}roblems for the
  {O}scillator algebra, the {L}ie {A}lgebra $\mathfrak{sl}_n$ and
  {M}ultivariate {K}rawtchouk {P}olynomials.''
  \href{http://dx.doi.org/10.1007/s00023-020-00972-8}{{\em Annales Henri
  Poincar{\'e}} {\bfseries 21} no.~12, (2020) 3939--3971}.

\bibitem{latini2020embedding}
D.~Latini, I.~Marquette, and Y.-Z. Zhang, ``{E}mbedding of the {R}acah algebra
  {R}(n) and superintegrability.''
  \href{http://dx.doi.org/https://doi.org/10.1016/j.aop.2021.168397}{{\em
  Annals of Physics} {\bfseries 426} (2021) 168397}.

\bibitem{PhysRevA.41.5666}
N.~W. Evans, ``{S}uperintegrability in classical mechanics.''
  \href{http://dx.doi.org/10.1103/PhysRevA.41.5666}{{\em Phys. Rev. A}
  {\bfseries 41} (1990) 5666--5676}.

\bibitem{FRIS1965354}
J.~Fri{\v{s}}, V.~Mandrosov, Y.~Smorodinsky, M.~Uhl{\'{\i}}{\v{r}}, and
  P.~Winternitz, ``On higher symmetries in quantum mechanics.''
  \href{http://dx.doi.org/10.1016/0031-9163(65)90885-1}{{\em Physics Letters}
  {\bfseries 16} no.~3, (1965) 354 -- 356}.

\bibitem{Makarov1967}
K.~V. A.~A.~Makarov, J. A.~Smorodinsky and P.~Winternitz, ``A systematic search
  for nonrelativistic systems with dynamical symmetries.''
  \href{http://dx.doi.org/10.1007/BF02755212}{{\em Il Nuovo Cimento A
  (1971-1996)} {\bfseries 52} no.~4, (1967) 1061--1084}.

\bibitem{EVANS1990483}
N.~Evans, ``{S}uper-integrability of the {W}internitz system.''
  \href{http://dx.doi.org/10.1016/0375-9601(90)90611-Q}{{\em Physics Letters A}
  {\bfseries 147} no.~8, (1990) 483 -- 486}.

\bibitem{doi:10.1063/1.529449}
N.~W. Evans, ``Group theory of the {S}morodinsky-{W}internitz system.''
  \href{http://dx.doi.org/10.1063/1.529449}{{\em Journal of Mathematical
  Physics} {\bfseries 32} no.~12, (1991) 3369--3375}.

\bibitem{Ballesteros_2004}
{\'A}.~Ballesteros, F.~Herranz, F.~Musso, and O.~Ragnisco, ``{S}uperintegrable
  deformations of the {S}morodinsky–{W}internitz {H}amiltonian.''
  \href{http://dx.doi.org/10.1090/crmp/037/01}{{\em CRM Proceedings and Lecture
  Notes} (2004) 1–14}.

\bibitem{doi:10.1063/1.2840465}
P.~E. Verrier and N.~W. Evans, ``A new superintegrable {H}amiltonian.''
  \href{http://dx.doi.org/10.1063/1.2840465}{{\em Journal of Mathematical
  Physics} {\bfseries 49} no.~2, (2008) 022902}.

\bibitem{1751-8121-42-24-245203}
{\'A}.~Ballesteros and F.~J. Herranz, ``{M}aximal superintegrability of the
  generalized {K}epler-{C}oulomb system on {N}-dimensional curved spaces.''
  \href{http://dx.doi.org/10.1088/1751-8113/42/24/245203}{{\em Journal of
  Physics A: Mathematical and Theoretical} {\bfseries 42} no.~24, (2009)
  245203}.

\bibitem{2011SIGMA...7..054T}
Y.~{Tanoudis} and C.~{Daskaloyannis}, ``{Algebraic Calculation of the Energy
  Eigenvalues for the Nondegenerate Three-Dimensional Kepler-Coulomb
  Potential}.'' \href{http://dx.doi.org/10.3842/SIGMA.2011.054}{{\em SIGMA}
  {\bfseries 7} (2011) 054}.

\bibitem{Ballesteros1996}
{ {\'A} Ballesteros, M. Corsetti and O. Ragnisco}, ``${N}$-dimensional
  classical integrable systems from {H}opf algebras.''
  \href{http://dx.doi.org/10.1007/BF01690329}{{\em Czechoslovak Journal of
  Physics} {\bfseries 46} no.~12, (1996) 1153--1163}.

\bibitem{0305-4470-31-16-009}
{\'A}.~Ballesteros and O.~Ragnisco, ``A systematic construction of completely
  integrable {H}amiltonians from coalgebras.''
  \href{http://dx.doi.org/10.1088/0305-4470/31/16/009}{{\em Journal of Physics
  A: Mathematical and General} {\bfseries 31} no.~16, (1998) 3791}.

\bibitem{1742-6596-175-1-012004}
{\'A}.~Ballesteros, A.~Blasco, F.~J. Herranz, F.~Musso, and O.~Ragnisco,
  ``({S}uper)integrability from coalgebra symmetry: {F}ormalism and
  applications.'' \href{http://dx.doi.org/10.1088/1742-6596/175/1/012004}{{\em
  Journal of Physics: Conference Series} {\bfseries 175} no.~1, (2009) 012004}.

\bibitem{2008PhyD..237..505B}
{\'A}.~{Ballesteros}, A.~{Enciso}, F.~J. {Herranz}, and O.~{Ragnisco}, ``{A
  maximally superintegrable system on an n-dimensional space of nonconstant
  curvature}.'' \href{http://dx.doi.org/10.1016/j.physd.2007.09.021}{{\em
  Physica D Nonlinear Phenomena} {\bfseries 237} no.~4, (Apr., 2008) 505--509}.

\bibitem{Ballesteros_2008}
{\'{A}}.~Ballesteros and A.~Blasco, ``$n$-dimensional superintegrable systems
  from symplectic realizations of lie coalgebras.''
  \href{http://dx.doi.org/10.1088/1751-8113/41/30/304028}{{\em Journal of
  Physics A: Mathematical and Theoretical} {\bfseries 41} no.~30, (Jul, 2008)
  304028}.

\bibitem{BALLESTEROS2011}
{\'A}.~{Ballesteros}, A.~{Enciso}, F.~J. {Herranz}, O.~{Ragnisco}, and
  D.~{Riglioni}, ``Quantum mechanics on spaces of nonconstant curvature: The
  oscillator problem and superintegrability.''
  \href{http://dx.doi.org/https://doi.org/10.1016/j.aop.2011.03.002}{{\em
  Annals of Physics} {\bfseries 326} no.~8, (2011) 2053 -- 2073}.

\bibitem{Ballesteros2009}
{\'A}.~{Ballesteros}, A.~{Enciso}, F.~J. {Herranz}, and O.~{Ragnisco},
  ``{S}uperintegrability on {N}-dimensional curved spaces: {C}entral
  potentials, centrifugal terms and monopoles.''
  \href{http://dx.doi.org/10.1016/j.aop.2009.03.001}{{\em Annals of Physics}
  {\bfseries 324} no.~6, (Jun, 2009) 1219–1233}.

\bibitem{Post_2015}
S.~Post and D.~Riglioni, ``Quantum integrals from coalgebra structure.''
  \href{http://dx.doi.org/10.1088/1751-8113/48/7/075205}{{\em Journal of
  Physics A: Mathematical and Theoretical} {\bfseries 48} no.~7, (2015)
  075205}.

\bibitem{Riglioni_2014}
D.~Riglioni, O.~Gingras, and P.~Winternitz, ``{S}uperintegrable systems with
  spin induced by co-algebra symmetry.''
  \href{http://dx.doi.org/10.1088/1751-8113/47/12/122002}{{\em Journal of
  Physics A: Mathematical and Theoretical} {\bfseries 47} no.~12, (2014)
  122002}.

\bibitem{LATINI20163445}
D.~Latini and D.~Riglioni, ``{F}rom ordinary to discrete quantum mechanics:
  {T}he {C}harlier oscillator and its coalgebra symmetry.''
  \href{http://dx.doi.org/https://doi.org/10.1016/j.physleta.2016.08.047}{{\em
  Physics Letters A} {\bfseries 380} no.~42, (2016) 3445 -- 3453}.

\bibitem{Latini_2019}
D.~Latini, ``{U}niversal chain structure of quadratic algebras for
  superintegrable systems with coalgebra symmetry.''
  \href{http://dx.doi.org/10.1088/1751-8121/aaffec}{{\em Journal of Physics A:
  Mathematical and Theoretical} {\bfseries 52} no.~12, (2019) 125202}.

\bibitem{10.1088/1751-8121/aac111}
Y.~Liao, I.~Marquette, and Y.-Z. Zhang, ``Quantum superintegrable system with a
  novel chain structure of quadratic algebras.''
  \href{http://dx.doi.org/10.1088/1751-8121/aac111}{{\em Journal of Physics A:
  Mathematical and Theoretical} {\bfseries 51} no.~25, (2018) 255201}.

\bibitem{doi:10.1119/1.9745}
H.~Goldstein, ``Prehistory of the \textquotedblleft
  {R}unge-{L}enz\textquotedblright vector.''
  \href{http://dx.doi.org/10.1119/1.9745}{{\em American Journal of Physics}
  {\bfseries 43} no.~8, (1975) 737--738}.

\bibitem{doi:10.1119/1.10202}
H.~Goldstein, ``More on the prehistory of the {L}aplace or {R}unge-{L}enz
  vector.'' \href{http://dx.doi.org/10.1119/1.10202}{{\em American Journal of
  Physics} {\bfseries 44} no.~11, (1976) 1123--1124}.

\bibitem{goldstein2002classical}
H.~Goldstein, C.~Poole, and J.~Safko, {\em Classical Mechanics}.
\newblock Addison Wesley, 2002.

\bibitem{Kuru2020}
S.~Kuru, I.~Marquette, and J.~Negro, ``{T}he general {R}acah algebra as the
  symmetry algebra of generic systems on pseudo-spheres.''
  \href{http://dx.doi.org/10.1088/1751-8121/abadb7}{{\em Journal of Physics A:
  Mathematical and Theoretical} {\bfseries 53} no.~40, (Sep, 2020) 405203}.

\bibitem{DeBie}
H.~De~Bie and W.~van~de Vijver, ``{A} {D}iscrete {R}ealization of the {H}igher
  {R}ank {R}acah {A}lgebra.''
  \href{http://dx.doi.org/10.1007/s00365-019-09475-0}{{\em Constructive
  Approximation} {\bfseries 52} no.~1, (2020) 1--29}.

\end{thebibliography}\endgroup

\end{document}